\documentclass[aps,pra,superscriptaddress,12pt,longbibliography]{revtex4-2}

\makeatletter
\renewcommand{\p@subsection}{}
\renewcommand{\p@subsubsection}{}
\makeatother

\usepackage{amsmath,amsfonts,amssymb,amsthm}
\usepackage{mathtools}
\usepackage{setspace}
\usepackage[normalem]{ulem} 
\usepackage{stmaryrd}       
\usepackage{dsfont}         
\usepackage{expdlist}
\usepackage{physics}
\usepackage{verbatim}
\usepackage{xspace}
\usepackage{subfigure}
\usepackage{tikz} 

\usepackage{graphicx,xcolor}
\usepackage{hyperref}
\hypersetup{pdfstartview=}
\usepackage{url}

\usepackage{verbatim}
\usepackage{alltt}
\usepackage{moreverb}

\usepackage{xr} 

\newtheorem{fact}{Fact}[section]
\newcommand{\rvline}{\hspace*{-\arraycolsep}\vline\hspace*{-\arraycolsep}}

\newtheorem{remark}{Remark}[section]
\newtheorem{example}{Example}[section]

\newtheorem{casse}{Case}

\newtheorem{prop}{Proposition}[section]

\providecommand{\ignore}[1]{}

\newif\ifcmnt
\cmnttrue
\ifdefined\cmntsoff\cmntfalse\fi

\ifcmnt

\else

\fi

\newcommand{\cC}{\mathcal{C}}
\newcommand{\cD}{\mathcal{D}}

\newcommand{\cH}{\mathcal{H}}

\newcommand{\cP}{\mathcal{P}}

\numberwithin{equation}{section}
\newtheorem{theorem}{Theorem}[section]
\newtheorem*{theorem*}{Theorem}
\newtheorem{problem}{Problem}[section]
\newtheorem*{problem*}{Problem}
\newtheorem{lemma}[theorem]{Lemma}
\newtheorem*{lemma*}{Lemma}
\newtheorem{corollary}[theorem]{Corollary}
\newtheorem{definition}[theorem]{Definition}

\begin{document}

\title{A Result About the Classification of Quantum Covariance Matrices Based on Their Eigenspectra}

\author{Arik Avagyan}
\email[]{arav2492@colorado.edu}
\affiliation{Department of Physics, University of Colorado, Boulder, Colorado, 80309, USA}

\begin{abstract} 

The set of covariance matrices of a continuous-variable quantum system with a finite number of degrees of freedom is a strict subset of the set of real positive-definite matrices due to Heisenberg's uncertainty principle. This has the implication that, in general, not every orthogonal transform of a quantum covariance matrix produces a positive-definite matrix that obeys the uncertainty principle. A natural question thus arises, to find the set of quantum covariance matrices consistent with a given eigenspectrum. For the special class of pure Gaussian states the set of quantum covariance matrices with a given eigenspectrum consists of a single orbit of the action of the orthogonal symplectic group. The eigenspectrum of a covariance matrix of a state in this class is composed of pairs that each multiply to one. Our main contribution is finding a non-trivial class of eigenspectra with the property that the set of quantum covariance matrices corresponding to any eigenspectrum in this class are related by orthogonal symplectic transformations. We show that all non-degenerate eigenspectra with this property must belong to this class, and that the set of such eigenspectra coincides with the class of non-degenerate eigenspectra that identify the physically relevant thermal and squeezing parameters of a Gaussian state.

\end{abstract}

\maketitle

\section{Overview}\label{sec overview}

A quantum system is a continuous-variable system when it is associated with an infinite-dimensional separable Hilbert space and has observables with continuous eigenspectra. An important area of physics that concerns itself with the study of such systems is quantum optics. There, the quantized electromagnetic field is usually described by a finite collection of one-dimensional harmonic oscillators, or modes. The covariance matrix (CM) of such a system captures the second moment of the canonical position and momentum operators of the modes. For a classical system described by a finite set of position and momentum variables, any real positive-definite matrix (PDM) is a valid CM of a probability distribution on phase-space. In quantum physics, however, Heisenberg's uncertainty principle sets a fundamental constraint on the size of the second moment. For example, for a system occupying a single mode, where the diagonal elements of the CM $\Gamma$ are $2( \langle \hat{q}^2 \rangle-\langle \hat{q} \rangle^2)$ and $2( \langle \hat{p}^2 \rangle-\langle \hat{p} \rangle^2)$, and each off-diagonal element is  $ \langle \hat{q}\hat{p}+\hat{p}\hat{q} \rangle-   2 \langle \hat{q} \rangle \langle \hat{p} \rangle$, the uncertainty principle can be stated as $\det(\Gamma) \geq 1$ in appropriate units.

The literature studying the various properties of quantum CMs is very rich. For one, quantum CMs are intrinsically interesting as they can be used to express such a fundamental tenet of quantum physics as the uncertainty principle. For example, Refs.~\cite{simon1994quantum,arvind1995real} present several equivalent formulations of the uncertainty-principle restrictions on the CM. Central to these formulations is Williamson's theorem, which states that every real PDM can be brought to a diagonal form by a symplectic transformation \cite{williamson1936on}. The diagonal entries of the resulting diagonal matrix, usually called the symplectic eigenvalues, are positive and appear in pairs. In the case of quantum CMs, the uncertainty principle sets a lower bound on the magnitudes of their symplectic eigenvalues. In the units adopted in this paper this bound is equal to $1$. It should be mentioned that the symplectic decomposition of real PDMs, in part motivated by its relevance to quantum physics, is an active field of study in mathematics, see, for example, Refs.~\cite{bhatia2015on, son2021computing,jain2022deriv}. Ref.~\cite{deGosson2001symplectic} is another example of the investigation of the relationship between the uncertainty principle and quantum CMs. There, the quantum CMs are represented as ellipsoids in phase space, where the uncertainty principle sets constraints on the sizes and shapes of these ellipsoids. The existence of these constraints is closely related to Gromov's non-squeezing theorem, a fundamental result in symplectic geometry \cite{gromov1985pseudo}.

The interest in quantum CMs increased with the advances in quantum information theory. In particular, quantum CMs are central to the field of Gaussian quantum information processing. A Gaussian state is completely characterized by its first moment and its CM, and quantum channels that preserve the Gaussianity of the state (called Gaussian channels or processes) are fully characterized by how they transform the first moment and the CM. For example, a unitary Gaussian process transforms the CM by a symplectic matrix, and thus preserves the symplectic eigenvalues. For Gaussian states, the other essential features of quantum systems besides the uncertainty principle, such as entanglement and separability, non-locality, as well as purity and entropy can be quantified based on particular properties of their CMs. Similarly, the CM is central to the evaluation of various Gaussian quantum information processing tasks, such as quantum communication, quantum cryptography, quantum computation, quantum teleportation and quantum state discrimination. \cite{weedbrook:qc2012a,serafini2017quantum}

The study presented here was also inspired by a particular application in quantum information processing. In a previous work \cite{avagyan2023multi} I and my coauthors showed that the total-photon-number distribution of a Gaussian state determines and is determined by the eigenspectrum of the CM and the absolute values of the first moment in each eigenspace. Thus, it might be possible to experimentally estimate these properties of an unknown Gaussian state using a device that counts the total photon-number, such as a transition-edge-sensor \cite{gerrits2016super}. In this context, it can be important to obtain estimates of more physically relevant properties of the Gaussian state from the measurements of the total-photon-number, such as bounds on the symplectic eigenvalues (usually called the thermal parameters in this field) of the CM, as well as bounds on the squeezing parameters of the CM. The latter are associated with the symplectic matrix in the symplectic decomposition of the CM, see Sec.~\ref{sec prob_form}.

Our main contribution in this paper is partially characterizing the class of eigenspectra of quantum CMs for which the thermal and squeezing parameters are exactly determined by the eigenspectrum (call it property $1$). This class includes the class of eigenspectra for which the set of CMs with the same eigenspectrum are related by orthogonal symplectic transformations (property $2$). It can be shown that one can construct a diagonal quantum CM from the eigenspectrum of any quantum CM, which we refer to as a ``diagonal representative". The thermal and squeezing parameters of a diagonal quantum CM are uniquely specified by its diagonal entries, or, more specifically, from the particular assignment of pairs of entries to the modes of the system. This implies that, for an eigenspectrum to possess property $1$, there must not exist another valid diagonal arrangement of the eigenspectrum that pairs the eigenvalues differently. We call this the ``general unique pairing condition". Part of our findings is that among the non-degenerate eigenspectra satisfying the general unique pairing condition, those that possess properties $1$ and $2$ are precisely the eigenspectra with diagonal representatives where at most one of the modes has a thermal parameter $\nu$ such that $\nu \geq 1$ in our units. For systems living in $S$ modes, we say the eigenspectra with the latter property satisfy the ``$(S-1)$-pure unique pairing condition". Our findings are not limited to non-degenerate eigenspectra. We narrow down the set of eigenspectra that need further study to finish the characterization of eigenspectra possessing either one or both of the properties $1$ and $2$.

Our main results are presented as a theorem, a proposition and a corollary of the proposition. Thm.~\ref{thm sub} shows that, given a non-degenerate eigenspectrum $\Lambda$ that satisfies the $(S-1)$-pure unique pairing condition, any two quantum CMs with $\Lambda$ are related by an orthogonal symplectic transformation. Such transformations preserve the thermal and squeezing parameters. Prop.~\ref{prop} shows that if the diagonal representative of an eigenspectrum $\Lambda$ has at least two thermal parameters $\nu_1$, $\nu_2$ such that $\nu_1>1$, $\nu_2>1$, and $\Lambda$ is non-degenerate in a particular subset of eigenvalues, then there  exists a quantum CM with $\Lambda$ that has a different set of thermal and squeezing parameters. Cor.~\ref{cor non-deg-class} shows that if an eigenspectrum is non-degenerate and possesses either property $1$ or property $2$, then it possesses both properties and satisfies the $(S-1)$-pure unique pairing condition.

The paper is organized as follows. Sec.~\ref{sec prob_form} is devoted to problem formulation. Sec.~\ref{sec ortho_prop} introduces a set of facts about orthogonal matrices. These facts are used in the proof of Thm.~\ref{thm sub} to construct a ``witness" for the violation of the uncertainty principle when an ONS transformation is applied to a diagonal representative of an eigenspectrum that satisfies the $(S-1)$-pure unique pairing condition. Sec.~\ref{sec graph_eig} uses graph theory to precisely define and explain the $(S-1)$-pure unique pairing condition. Thm.~\ref{thm sub}, Prop.~\ref{prop} and Cor.~\ref{cor non-deg-class} are presented in Sec.~\ref{sec main_result}. We close with a brief discussion of the implications of our results and proposals for future research in Sec.~\ref{sec disc}.

\section{Problem Formulation}\label{sec prob_form}

The background material
for this section can be found in a textbook such as Ref.~\cite{serafini2017quantum}. 
Consider a continuous variable system composed of $S$ modes characterized by canonical position and momentum operators $\hat{q}_i$, $\hat{p}_i$ for $i=1,\hdots,S$. We choose our units such that these operators satisfy the commutation relations $ [ \hat{q}_i,\hat{q}_j  ]=0 $, $ [ \hat{p}_i,\hat{p}_j  ]=0 $, and $[ \hat{q}_i,\hat{p}_j  ] = i \delta_{ij} $. These commutation relations can be expressed more compactly using the anti-symmetric $2$-form $\Omega = \bigoplus_{i=1}^S \begin{pmatrix} 0 & 1 \\ -1 & 0 \end{pmatrix}$. In particular, $ [\hat{r}_i,\hat{r}_j ] = i \Omega_{ij}$, where $\hat{r}_{2k-1}  =\hat{q}_{k}$ and $\hat{r}_{2k}  =\hat{p}_{k}$ for $k=1,\hdots,S$. The covariance matrix $\Gamma$ of the physical system is defined as $\Gamma_{ij} = \langle
\hat{r}_i \hat{r}_j + \hat{r}_j \hat{r}_i \rangle - 2 \langle
\hat{r}_i \rangle \langle \hat{r}_j \rangle$. The notation $\langle \ldots \rangle$ denotes the expectation with respect to the state of the system. We also introduce the phase-space coordinates $(q_1,p_1,\hdots,q_S,p_S)
$. We refer to each pair $(q_i,p_i)$ as conjugate coordinates. The basis vectors in this coordinate system are denoted by $\vec{e}_i$ for $i = 1,\hdots,2S$, e.g. $\vec{e}_1 = (1,0,\hdots,0)^T$, $\vec{e}_2 = (0,1,0,\hdots,0)^T$, and so on. In these coordinates the restriction imposed on $\Gamma$ by Heisenberg's uncertainty principle can be expressed as
\begin{align}\label{eq: uncer_prin}
\Gamma + i\Omega \geq 0.
\end{align}
We refer to the group of orthogonal symplectic matrices by $\mathrm{SpO}(2S,\mathbb{R})$. 

Using Williamson's theorem, every PDM, and, in particular, every quantum CM $\Gamma$ can be decomposed as $\Gamma = A^TTA$, where $A$ is symplectic, and $T$ is diagonal and consisting of $S$ blocks of the form $\begin{pmatrix}\nu_{i} & 0 \\
  0 & \nu_{i}\end{pmatrix}$. With our conventions, Eq.~\ref{eq: uncer_prin} implies that each $\nu_i$ must satisfy $\nu_{i}\geq 1$ for $\Gamma $ to be a quantum CM. $A$ satisfies $A \Omega A^T = \Omega$, and can be further decomposed as $A=LQK$, where $Q$ is diagonal and consisting of $S$ blocks of the form $\begin{pmatrix} e^{r_i} & 0 \\ 0 & e^{-r_i}  \end{pmatrix}$, and $K,L \in \mathrm{SpO}(2S,\mathbb{R})$ . The $\nu_i$ are the symplectic eigenvalues or the thermal parameters of $\Gamma$, while the $\abs{r_i}$ are referred to as the squeezing parameters of $\Gamma$. We reserve referring to the $\nu_i$ as thermal parameters when $\nu_i \geq 1$, so that they can be associated with a quantum CM. 
  
 We now introduce the notions of non-trivial and trivial ONS transformations.  
\begin{definition}\label{define non-trivial ONS}
Consider a quantum CM $\Gamma$ and an ONS matrix $O$. We say $O^T \Gamma O$ is a non-trivial ONS transformation of $\Gamma$ if $O$ cannot be written as $O=O'W$, where $W \in \mathrm{SpO}(2S,\mathbb{R})$ and $O'$ commutes with $\Gamma$. Conversely, we say $O^T \Gamma O$ is a trivial ONS transformation of $\Gamma$ if $O^T \Gamma O = W^T \Gamma W$ for some $W \in \mathrm{SpO}(2S,\mathbb{R})$.
\end{definition} 
In this paper we mainly focus on quantum CMs that have non-degenerate eigenspectra. For such matrices every ONS transformation is non-trivial since the diagonal matrix in the eigendecomposition has no duplicate diagonal entries. 

Since a general quantum CM $\Gamma$ has the decomposition $\Gamma = K^TQL^TTLQK$, any transformation $W^T\Gamma W$, where $W \in \mathrm{SpO}(2S,\mathbb{R})$, does not affect the thermal and squeezing parameters. However, a non-trivial ONS transformation usually results in a PDM with a different set of symplectic eigenvalues and squeezing parameters. If the new PDM has a symplectic eigenvalue $\nu$ such that $\nu <1$, then the PDM is not a quantum CM. For example, for the important class of CMs of pure Gaussian states, for which $T$ is the identity, any non-trivial ONS transformation results in a PDM that is not a quantum CM. The partial proof of this claim (for non-degenerate eigenspectra) is part of the proof of Thm.~\ref{thm sub}, but we provide an alternative proof that covers all such eigenspectra in App.~\ref{app}. The latter uses a more direct and a much simpler approach that might help the reader in forming a wider perspective on the arguments in Sec.~\ref{sec main_result}. More generally, there are many quantum CMs with different thermal and squeezing parameters that have the same eigenspectrum. These facts make the broad task of classifying quantum CMs based on their eigenspectra a challenging problem. In this paper we make significant progress in solving two related subproblems of this general problem. To formulate the problems, let $\Lambda = (\lambda_i)_{i=1}^{2S}$ denote the eigenspectrum of a quantum CM of a system living in $S$ modes.

\begin{problem}\label{problem 1}
Find the full set of conditions on $\Lambda$ for which the set of quantum CMs with this eigenspectrum consist of a single orbit under the action of $\mathrm{SpO}(2S,\mathbb{R})$.
\end{problem}
\begin{problem}\label{problem 2}
Find the full set of conditions on $\Lambda$ for which the set of quantum CMs with this eigenspectrum have the same thermal and squeezing parameters.
\end{problem}
Let us denote the (so far unknown) class of eigenspectra that satisfy the conditions of Prob.~\ref{problem 1} (Prob.~\ref{problem 2}) by $\cP_1(S)$ ($\cP_2(S)$). Notice that $\cP_1(S) \subseteq \cP_2(S)$ as for any $\Lambda \in \cP_1(S)$ the corresponding set of CMs are related to each other by transformations in $\mathrm{SpO}(2S,\mathbb{R})$ that preserve the thermal and squeezing parameters. 

We need the following definitions and observations. Let us denote the symmetric group on $n$ letters by $\mathrm{Sym}(n)$. 
\begin{definition}\label{define cD_Lambda}
Consider an eigenspectrum $\Lambda=  (\lambda_i)_{i=1}^{2S}$ of a quantum CM $\Gamma$, where the $\lambda_i$ are in a non-ascending order. We refer to any matrix in the set \[ \cD_\Lambda = \left\{ P \bigoplus_{i=1}^S \begin{pmatrix} \lambda_{i} & 0 \\ 0 & \lambda_{2S+1-i} \end{pmatrix}  P^T\right \}_{P \in  \mathrm{SpO}(2S,\mathbb{R}) \cap \mathrm{Sym}(2S) } \] as a diagonal representative of $\Lambda$ or of $\Gamma$. 
\end{definition}
To shed light on the composition of $\cD_\Lambda$, note that each $D \in \cD_\Lambda$ is diagonal, the diagonal entries being a symplectic permutation of the sequence $(\lambda_1,\lambda_{2S},\lambda_{2},\lambda_{2S-1},\hdots, \lambda_{S},\lambda_{S+1} ) $. Symplectic permutations consist of all permutations that preserve the pairings between the elements at the locations $2i-1,2i$ for all $i = 1,\hdots, S$ in the sequence. By this we mean that if $\lambda_i$ is moved to the location $2j-1$ or $2j$ for some $j = 1,\hdots, S$, then $\lambda_{2S-i+1}$ is moved to the location $2j$ or $2j-1$, respectively. Importantly, given a quantum CM $\Gamma$, all of its diagonal representatives are valid quantum CMs. For a proof see, for example, Lem.~4.3 in Ref.~\cite{avagyan2023multi}. 
\begin{remark}\label{remark Lambda_in_P_1}
Consider a quantum CM $\Gamma$ with an eigenspectrum $\Lambda \in \cP_1(S)$. Then, for any $D \in \cD_\Lambda $ there exists a $K \in \mathrm{SpO}(2S,\mathbb{R})$ such that $\Gamma = K ^T D K$.
\end{remark}
Further, note that the symplectic decomposition of any $ D \in \cD_\Lambda$, and, indeed, of any diagonal quantum CM, is of the form $D= Q^2 T$. Therefore, the squeezing and thermal parameters of $ D $ are identified by the diagonal entries. We can be more specific.
\begin{fact}\label{remark u_and_r_det_diag}
Let $D = \mathrm{diag} [d_1,\hdots,d_{2S}]$ be a diagonal quantum CM. The set of thermal and squeezing parameters of $D$ are given by $\{\nu_i\}_{i=1}^S = \{\sqrt{d_{2i-1}d_{2i} } \}_{i=1}^S$ and $\{\abs{r_i} \}_{i=1}^S = \{ \ln (d_{2i}/d_{2i-1})/4 \}_{i=1}^S$. If $D \in \cD_\Lambda$ for some eigenspectrum $\Lambda$, then $\{\nu_i\}_{i=1}^S= \{\sqrt{\lambda_i \lambda_{2S+1-i}}  \}_{i=1}^S$ and $\{ \abs{r_i}\}_{i=1}^S =\{ \ln (\lambda_i/\lambda_{2S+1-i})/4 \}_{i=1}^S$.
\end{fact}
The first part of Fact~\ref{remark u_and_r_det_diag} follows from the definitions of the $\nu_i$ and the $\abs{r_i} $. The second part implies that all matrices in $ \cD_\Lambda$ have the same set of thermal and squeezing parameters. This follows from the construction of $\cD_\Lambda$, namely, from the fact that the matrices in $ \cD_\Lambda$ are related by symplectic permutations that do not affect the thermal and squeezing parameters.

The final observation is a particular consequence of Eq.~\ref{eq: uncer_prin} that we use in the construction of the uncertainty-principle-violating witness in the proof of Thm.~\ref{thm sub}. Let us denote the diagonal consecutive $2 \times 2$ blocks of a matrix $A$ by $A^{(i)}$. More specifically, $A^{(i)} = \begin{pmatrix} A_{2i-1,2i-1} & A_{2i-1,2i}  \\ A_{2i,2i-1} & A_{2i,2i} \end{pmatrix}$.
\begin{fact}\label{fact det_Gamma_i}
Consider a quantum CM $\Gamma$. Then, $\det (\Gamma^{(i)}) \geq 1$ for all $i = 1,\hdots, S$.
\end{fact}
To see this, first note that $\Gamma + i\Omega \geq 0$ implies $\Gamma^{(i)} +i\Omega_1 \geq 0$ for every $i = 1,\hdots, S$, where $\Omega_1 = \begin{pmatrix} 0 & 1 \\ -1 & 0 \end{pmatrix}$. This in turn implies that $\det (\Gamma^{(i)} +i\Omega_1) \geq 0$. The fact statement then follows since the $\Gamma^{(i)}$ are symmetric. To be more specific, 
\begin{align}
&  \det \begin{pmatrix} \Gamma_{2i-1,2i-1}      & \Gamma_{2i-1,2i} + i \\   \Gamma_{2i-1,2i} - i & \Gamma_{2i,2i} \end{pmatrix} \geq 0  \implies    \Gamma_{2i-1,2i-1}     \Gamma_{2i,2i}  - (\Gamma_{2i-1,2i} + i ) (\Gamma_{2i-1,2i} - i ) \geq 0  \nonumber \\
&\implies  \Gamma_{2i-1,2i-1}     \Gamma_{2i,2i}  - \Gamma_{2i-1,2i} \Gamma_{2i-1,2i} -1 \geq 0  \implies \det \Gamma^{(i)} -1 \geq 0.
\end{align}
 
 \section{On Some Properties of Orthogonal Matrices}\label{sec ortho_prop}

Here we state or derive several properties of orthogonal matrices that we use in the proof of Thm.~\ref{thm sub}. We denote the orthogonal group in $n$ dimensions by $\mathrm{O}(n)$.

\begin{fact}\label{O_det_conv} 
Consider an arbitrary $O \in \mathrm{O}(n)$ and an arbitrary set of $p$ indices $i_1 <\hdots < i_p$, where $i_p \leq n$. Denote the submatrix of $O$ that is produced from the columns $i_1,\hdots,i_p$ and rows $j_1,\hdots,j_p$ of $O$ by $O_{j_1,\hdots,j_p}^{i_1,\hdots,i_p}$. Then,
\begin{equation}
\sum_{j_1 <\hdots < j_p} \left[ \det (O_{j_1,\hdots,j_p}^{i_1,\hdots,i_p}) \right]^2 =1
\end{equation}
\end{fact}
\begin{proof}
This fact is a straightforward application of the Cauchy-Binet formula \cite[Chap.~3]{taotopics} to the product of any matrix $A$ composed of mutually orthonormal columns with its transpose. More specifically, assume $A$ has dimensions $n \times m$, $n \geq m$, without loss of generality. Let $[n ]$ denote the set $\{1,\hdots,n\}$, and let $\binom{[n]}{m}$ denote the set of subsets of $[n]$ of size $m$. Also, for $x \in  \binom{[n]}{m}$ let $A_{x,[m]}$ denote the square matrix composed of the rows of $A$ at indices from $x$ and let $(A^T)_{[m],x}$ denote the square matrix composed of the columns of $A^T$ at indices from $x$. Note that $(A^T)_{[m],x} = (A_{x,[m]})^T$. Then, the Cauchy-Binet formula states
\begin{align}
 \det(A^TA) &= \sum_{x \in \binom{[n]}{m}}\det((A^T)_{[m],x})\det(A_{x,[m]}) \nonumber \\
& =\sum_{x \in \binom{[n]}{m}}\det(A_{x,[m]})^2,
\end{align}
where we have used a property of the determinant to equate $\det((A^T)_{[m],x}) = \det(A_{x,[m]})$ in the second line.
It is left to note that $A^TA = I$, hence $\det(A^TA) =1$, and that $A$ can be chosen to be the matrix composed of the columns $i_1,\hdots,i_p$ of $O$. 

\end{proof}
 This fact implies that the squares of the determinants of all square submatrices in any column set of $O$ can be thought of as a finite set of weights that sum to $1$. By transposition the same is true for any row set. We use the notation for submatrices in the fact statement throughout the paper. 

For the last several properties presented in this section we need the fact that $\mathrm{SpO}(2n,\mathbb{R})$ is isomorphic to a representation of the unitary group in $n$ dimensions, $\mathrm{U}(n)$ \cite{arvind1995real}. To better visualize this relationship, we order the coordinates as $(q_1,\hdots,q_n,p_1,\hdots,p_n)$ in this paragraph. Then, the matrix $W \in \mathrm{SpO}(2n,\mathbb{R})$ corresponding to the unitary $U \in \mathrm{U}(n)$ can be represented by
\begin{align}\label{U_in_SpO}
W &= \frac{1}{2} \begin{pmatrix} U+U^* & \rvline & i(U-U^*) \\ \hline  i(U^*-U) & \rvline &  U+U^* \end{pmatrix} \nonumber \\
& = \begin{pmatrix} \mathrm{Re}(U) & \rvline & - \mathrm{Im}(U) \\ \hline  \mathrm{Im}(U) & \rvline &  \mathrm{Re}(U)\end{pmatrix}
\end{align}
 Eq.~\ref{U_in_SpO} is obtained by conjugating $\begin{pmatrix} U & \rvline & 0 \\ \hline  0 & \rvline &  U^* \end{pmatrix}$ with the unitary $\frac{1}{\sqrt{2}} \begin{pmatrix} I_n & \rvline & i I_n \\ \hline  I_n & \rvline &  -i I_n \end{pmatrix}$, where $I_n$ is the $n \times n$ identity. In the rest of this section, as well as in the rest of the paper we stick to ordering the phase-space coordinates as $(q_1,p_1,\hdots,q_n,p_n)$. 

We introduce the maps $c_i[\cdot]$ and $r_i[\cdot]$, which, for any $n\times n$ matrix $A$, where $n$ is arbitrary, return its $i$'th column, $c_i[A] = (A_{1i},\hdots,A_{ni})^T$ and its $i$'th row, $r_i[A] = (A_{i1},\hdots,A_{in})$, respectively. The discussion in the previous paragraph leads to the following fact.
\begin{fact}\label{fact 2}
Given any $W \in \mathrm{SpO}(2n,\mathbb{R})$, for any $1 \leq i \leq n$, the columns $c_{2i-1}[W]$ and $c_{2i}[W]$, and the rows $r_{2i-1}[W]$ and $r_{2i}[W]$ can be expressed as
\begin{align}\label{col_OS_exp}
c_{2i-1}[W] & = (\mathrm{Re}(u_1),\mathrm{Im}(u_1),\hdots,\mathrm{Re}(u_n),\mathrm{Im}(u_n)) \nonumber \\
c_{2i}[W] & = (-\mathrm{Im}(u_1),\mathrm{Re}(u_1),\hdots,-\mathrm{Im}(u_n),\mathrm{Re}(u_n)) \nonumber \\
r_{2i-1}[W] & = (\mathrm{Re}(u_1), -\mathrm{Im}(u_1),\hdots,\mathrm{Re}(u_n),-\mathrm{Im}(u_n)) \nonumber \\
r_{2i}[W] & = (\mathrm{Im}(u_1),\mathrm{Re}(u_1),\hdots,\mathrm{Im}(u_n),\mathrm{Re}(u_n))
\end{align}
where the $u_i$ are complex numbers such that $\sum_{i=1}^n \abs{u_i}^2 = 1$. Conversely, given any column and row set of this form one can find a matrix in $  \mathrm{SpO}(2n,\mathbb{R})$ to which they belong. 
\end{fact}

\begin{lemma}\label{lem col12_O_rep}
Consider an arbitrary $O \in \mathrm{O}(2n)$. Then, for any $1 \leq i \leq n$ and $1 \leq j \leq 2n$, and for any real constant $\lambda $ there exists a matrix $W \in \mathrm{SpO}(2n,\mathbb{R})$ such that $(c_{2i-1}[OW]))_k = \lambda (c_{2i}[OW])_k$ for all $k \neq j$. Consequently, $\det( (OW)^{2i-1,2i}_{kl}) = 0$ for all $k<l$ where $k \neq j$, $l \neq j$.
\end{lemma}
\begin{proof}
$c_{2i-1}[OW]$ and $c_{2i}[OW]$ are equal to $Oc_{2i-1}[W]$ and $Oc_{2i}[W]$, respectively. Let us denote by $\tilde O$ the matrix produced by removing row $j$ from $O$. We want to show that one can find a matrix $W \in \mathrm{SpO}(2n,\mathbb{R}) $ such that $\tilde O c_{2i-1}[W] = \lambda \tilde O c_{2i}[W] $. This implies that $c_{2i-1}[W]-\lambda c_{2i}[W]$ lies in $\mathrm{ker} (\tilde O )$. $\mathrm{ker}(\tilde O)$ is one-dimensional since it is obtained by removing a row from a full rank square matrix, and, by the arbitrariness of $O$, it can be the span of any vector in $\mathbb{R}^{2n}$. Thus, we need to show that given any one-dimensional subspace of $\mathbb{R}^{2n}$ there exists a matrix $W \in \mathrm{SpO}(2n,\mathbb{R})$ such that $c_{2i-1}[W]-\lambda c_{2i}[W]$ spans this subspace. To do this, we substitute Eq.~\ref{col_OS_exp} to obtain
\begin{align}
& c_{2i-1}[W]-\lambda c_{2i}[W]  = \nonumber \\
& (\mathrm{Re}(u_1)+\lambda \mathrm{Im}(u_1), \mathrm{Im}(u_1)- \lambda \mathrm{Re}(u_1),\hdots,\mathrm{Re}(u_n)+\lambda \mathrm{Im}(u_n), \mathrm{Im}(u_n)- \lambda \mathrm{Re}(u_n)) \nonumber \\
& =  \begin{pmatrix} 1 & \lambda \\ - \lambda & 1  \end{pmatrix}^{\bigoplus n} c_{2i-1}[W].
\end{align}
It can be seen that $c_{2i-1}[W]-\lambda c_{2i}[W]$ is equal to an invertible matrix times $c_{2i-1}[W]$ for any value of $\lambda$. Since $c_{2i-1}[W]$ can be chosen to point in any direction in $\mathbb{R}^{2n}$, we can choose this direction such that $c_{2i-1}[W]-\lambda c_{2i}[W] $ lies in the desired subspace.

\end{proof}

\begin{lemma}\label{lem col12_O_rep_3}
Consider an arbitrary $O \in \mathrm{O}(2n)$. Then, for any $1 \leq i \leq n$, and for any $n-1$ rows $j_1,\hdots,j_{n-1}$ there exists a matrix $W \in \mathrm{SpO}(2n,\mathbb{R})$ such that these rows in $c_{2i-1}[OW]$ and $c_{2i}[OW]$ are zero. Subsequently, $\det( (OW)^{2i-1,2i}_{kl}) = 0$ for all $k<l$ where either $k \in \{j_1,\hdots,j_{n-1} \}$ or $l \in  \{ j_1,\hdots,j_{n-1}\}$.
\end{lemma}
\begin{proof}
As stated in the proof of the previous lemma $c_{2i-1}[OW]$ and $c_{2i}[OW]$ are equal to $Oc_{2i-1}[W]$ and $Oc_{2i}[W]$, respectively. The claim is that there exists a matrix $W \in \mathrm{SpO}(2n,\mathbb{R})$ such that $r_{l}[O] \cdot c_{2i-1}[W] =0$ and $r_{l}[O] \cdot c_{2i}[W] =0$ for all $l \in \{j_1,\hdots,j_{n-1} \}$. The second set of equalities are equivalent to
\begin{equation}
(-O_{l,2},O_{l,1},\hdots,-O_{l,2n},O_{l,2n-1})\cdot c_{2i-1}[W] =0.
\end{equation}
This can be seen from the relationship between $c_{2i}[W] $ and $c_{2i-1}[W] $ in Eq.~\ref{col_OS_exp}. Thus, the requirement on $W$ is that $c_{2i-1}[W] $ is in the kernel of a $(2n-2) \times 2n$ matrix which is obtained by stacking the $n-1$ rows $\{r_{l}[O] \}_{ l = j_1,\hdots,j_{n-1}  }$ vertically on top of the $n-1$ rows $\left \{(-O_{l,2},O_{l,1},\hdots,-O_{l,2n},O_{l,2n-1})\right \}_{ l = j_1,\hdots,j_{n-1}  }$. This is always possible since $c_{2i-1}[W]$ can be chosen to be any unit vector in $\mathbb{R}^{2n}$.
\end{proof}
The following statement is self-evident, but laying it out helps in conveying the significance of the definition that follows.
\begin{lemma}\label{lem OW=1_oplus_O'}
Consider an arbitrary $O \in \mathrm{O}(2n)$. There exists $W \in \mathrm{SpO}(2n,\mathbb{R})$ such that $OW = I_1 \oplus O'$ for some $O' \in \mathrm{O}(2n-1)$. 
\end{lemma}
\begin{proof}
We can choose $W$ such that $c_1[W] = r_1[O]$. Then, $OW$ must have the desired form due to the orthonormality of its columns and its rows.
\end{proof}
 
\begin{definition}\label{define max_right_red_trans}
Consider an ONS matrix $O \in \mathrm{O}(2S)$. For each $1 \leq k \leq S$, let $\cH_k(O) \subset  \mathrm{SpO}(2S,\mathbb{R})$ denote the set of all $W \in \mathrm{SpO}(2S,\mathbb{R})$ such that $OW = I_{2S-2k} \oplus O'$, where $O' \in  \mathrm{O}(2k)$, $O' \neq I_2 \oplus O''$ for some $O'' \in \mathrm{O}(2k-2)$. Denote $n = \min\{k   \mid \cH_k(O) \neq \emptyset \}$. Given any $W \in \cH_n(O)$, we call $OW$ a ``maximally mode-reduced transform of $O$", and we call $n$ the ``maximally reduced mode number of $O$". 
\end{definition}
Notice that, given a maximally mode-reduced transform of an ONS matrix $O$, $ I_{2S-2n} \oplus O'$, where $n$ is the maximally reduced mode number of $O$, we can always find $W' \in \mathrm{SpO}(2S,\mathbb{R})$ such that $OW' = I_{2S-2n+1} \oplus O''$, where $O'' \in  \mathrm{O}(2n-1)  \setminus (\mathrm{O}(1) \oplus \mathrm{O}(2n-2)) $. This is possible by Lem.~\ref{lem OW=1_oplus_O'}. 
\begin{lemma}\label{lem n_geq_2}
Consider an arbitrary ONS matrix $O \in \mathrm{O}(2S)$. Then, $\cH_1(O)$ is empty, and hence the maximally reduced mode number $n$ satisfies $n \geq 2$. 
\end{lemma}
\begin{proof}
The claim is that there does not exist $W \in \mathrm{SpO}(2n,\mathbb{R})$ such that $OW = I_{2n-2} \oplus O'$ for some $O' \in \mathrm{O}(2)$. 
Assume the opposite. Then, $OW \in \mathrm{SpO}(2n,\mathbb{R})$, and hence $ O^T = W ( I_{2n-2} \oplus (O')^T) \in \mathrm{SpO}(2n,\mathbb{R})$. This contradicts the assumption that $O$ is an ONS matrix.
\end{proof}
\begin{lemma}\label{lem det=1_impos}
Consider an arbitrary ONS matrix $O \in \mathrm{O}(2S)$. Let $I_{2S-2n} \oplus O'$ be a maximally mode-reduced transform of $O$, with the maximally reduced mode number $n$. Denote $i_n = S-n+1 $. Then, there does not exist $W \in \mathrm{SpO}(2n,\mathbb{R})$ such that 
\begin{align}\label{eq lem det=1_impos}
\left[ \det \left(  I_{2S-2n} \oplus O'W  \right)_{2i_n-1,2i_n}^{2i_n-1,2i_n} \right]^2 = 1.
\end{align}
\end{lemma}
\begin{proof}
Assume the opposite, and let $W \in \mathrm{SpO}(2n,\mathbb{R})$ satisfy Eq.~\ref{eq lem det=1_impos}. Denote $A =I_{2S-2n} \oplus O'W$. Then, it follows from Fact~\ref{O_det_conv} that $ \det A_{k,l}^{2i_n-1,2i_n} =0$, $ \det  A^{k,l}_{2i_n-1,2i_n} =0$ for all $(k,l) \neq (2i_n-1,2i_n)$. This is only possible if $A_{2i_n-1,l}=0, A_{2i_n,l} =0$ for all $l \neq 2i_n-1,2i_n$, and, similarly, if $A_{k,2i_n-1}=0, A_{k,2i_n} =0$ for all $k \neq 2i_n-1,2i_n$. But this means $A = I_{2S-2n} \oplus O_2 \oplus O_{2n-2}$ for some $O_2 \in \mathrm{O}(2)$ and $O_{2n-2} \in \mathrm{O}(2n-2)$. We can turn $A$ into $ I_{2S-2n-2} \oplus O_{2n-2}$ by right-multiplying with $I_{2S-2n} \oplus O^T_2 \oplus I_{2n-2} \in  \mathrm{SpO}(2S,\mathbb{R})$. This contradicts the assumption that $I_{2S-2n} \oplus O'$ is a maximally mode-reduced transform of $O$.

\end{proof}

\section{Graph-theoretic representation of eigenspectra }\label{sec graph_eig}
 
In this section we use tools from graph theory to describe certain properties of eigenspectra. This description simplifies the narration of the proof of Thm.~\ref{thm sub}, which is our main contribution. See, for example, Ref.~\cite{bondy1976graph} for the relevant background in graph theory.

 We associate with each eigenspectrum $\Lambda$ a graph $G(\Lambda)$. 
 \begin{definition}\label{define graph of eig}
Consider a non-degenerate eigenspectrum $\Lambda= (\lambda_i)_{i=1}^{2S}$, where the $\lambda_i$ are in a decreasing order. We refer to the graph with vertices
$\{\lambda_1,\ldots, \lambda_{2S}\}$ and edges $\{\{\lambda_i,\lambda_j\} \mid i<j, \lambda_{i} \lambda_{j}
\geq 1\}$ by $G(\Lambda)$. We say an edge $\{\lambda_i,\lambda_j\}$ of $G(\Lambda)$ is pure if $\lambda_{i}\lambda_{j}=1$. 
\end{definition} 
The most significant use of $G(\Lambda)$ is to capture the ``$(S-1)$-pure unique pairing condition" alluded to in Sec.~\ref{sec overview}.

 \begin{definition}\label{define t-pure}
Consider a non-degenerate eigenspectrum $\Lambda= (\lambda_i)_{i=1}^{2S}$, where the $\lambda_i$ are in a decreasing order. We say that $\Lambda$ satisfies the unique pairing condition if $G(\Lambda)$ has a unique perfect matching. For a nonnegative integer $t \neq S$, $\Lambda$ is said to satisfy the $t$-pure unique pairing condition if at least $t$ of the
edges of the perfect matching are pure.
\end{definition}

Thm.~\ref{thm sub} is about the non-degenerate $\Lambda$ satisfying the $(S-1)$-pure unique pairing condition, so we now aim to build an intuition about such $\Lambda$. First, we find the following definition useful.
\begin{definition}
Consider an eigenspectrum $\Lambda =  (\lambda_i)_{i=1}^{2S}$, where the $\lambda_i$ are in a non-ascending order. Consider a double-valued map $(n_1(k),n_2(k))$, $n_1(k) \neq n_2(k)$, assigning mutually distinct pairs of indices, with each index belonging to $\{1,\hdots,2S \}$, to every $k = 1,\hdots, S$. This map produces the following pairing of the eigenvalues: $\{  \{ \lambda_{n_1(k)},\lambda_{n_2(k)}  \} \}_{k=1}^S$. We say $\{  \{ \lambda_{n_1(k)},\lambda_{n_2(k)}  \} \}_{k=1}^S$ is a valid pairing of the $\lambda_i$ if $\lambda_{n_1(k)} \lambda_{n_2(k)} \geq 1$ for each $k = 1,\hdots,S$.
\end{definition}
The unique pairing condition for $\Lambda$, as the name implies, says that one can assign only one valid pairing to the $\lambda_i \in \Lambda$. This is captured by the requirement that a perfect matching of $G(\Lambda)$ exists and is unique. The $(S-1)$-pure unique pairing condition requires that there can exist one and only one pair $\{\lambda_i,\lambda_j\}$ in the valid pairing such that $ \lambda_{i} \lambda_{j} \geq 1$. The unique pairing condition has an implication for the set of diagonal quantum CMs with $\Lambda$. First, note that it follows from Fact~\ref{remark u_and_r_det_diag} that $\lambda_i \lambda_{2S+1-i} \geq 1$ for all $i = 1,\hdots, S$, since the set of thermal parameters of each $D \in \cD_\Lambda$ corresponds to $\{\nu_i\}_{i=1}^S= \{\sqrt{\lambda_i \lambda_{2S+1-i}}  \}_{i=1}^S$. Therefore, the unique perfect matching of $G(\Lambda)$, when $\Lambda$ satisfies the unique pairing condition, consists of the edges $\{ \{\lambda_i,\lambda_{i+2S-1} \} \}_{i=1}^S$. 

\begin{lemma}\label{lem P_2_and_diagrep}
Consider a non-degenerate eigenspectrum $\Lambda = (\lambda_i)_{i=1}^{2S}$, where the $\lambda_i$ are in a decreasing order, that satisfies the unique pairing condition. The set of diagonal quantum CMs with $\Lambda$ consists of $\cD_\Lambda$. 
\end{lemma}
\begin{proof}
Assume that there exists a diagonal quantum CM $D_0 = \mathrm{diag} [d_1,\hdots, d_{2S} ] $ with the eigenspectrum $\Lambda$ such that $D_0 \notin \cD_\Lambda$. Since $D_0$ is a quantum CM, it's thermal parameters $\nu_i = \sqrt{d_{2i-1} d_{2i} }$ satisfy $\nu_i \geq 1$. This implies that $\Lambda$ admits a valid pairing of its eigenvalues given by $\{\{d_{2i-1},d_{2i}\} \}_{i=1}^S$. This pairing must be different from $\{ \{ \lambda_i,\lambda_{2S+i-1} \}   \}_{i=1}^S$, since otherwise $D_0 \in \cD_\Lambda$, which contradicts our assumption. But then $\Lambda$ does not satisfy the unique pairing condition. The claim of the lemma then follows.
\end{proof}
\begin{corollary}
Consider a non-degenerate eigenspectrum $\Lambda = (\lambda_i)_{i=1}^{2S}$, where the $\lambda_i$ are in a decreasing order, that satisfies the $(S-1)$-pure unique pairing condition. Let $(\nu_i )_{i=1}^S$ denote the set of thermal parameters of a diagonal quantum CM $D$ with the eigenspectrum $\Lambda$, where the $\nu_i$ are in a non-decreasing order. Then, $\nu_i=1$ for $i = 1,\hdots, S-1$. 
\end{corollary}
\begin{proof}
According to Lem.~\ref{lem P_2_and_diagrep} $D \in \cD_\Lambda$, therefore, by Fact~\ref{remark u_and_r_det_diag}, $\{\nu_i\}_{i=1}^S= \{\sqrt{\lambda_i \lambda_{2S+1-i}}  \}_{i=1}^S$. The claim follows since, as discussed above, $\{\lambda_i ,\lambda_{2S+1-i} \}_{i=1}^S$ is the set of edges in the unique perfect matching associated with $\Lambda$. 
\end{proof}

For what follows we find it convenient to introduce a different method of labeling the eigenvalues.
\begin{definition}
Consider a non-degenerate $\Lambda = (\lambda_i)_{i=1}^{2S}$, where the $\lambda_i$ are in a decreasing order. Assume $\Lambda$ satisfies the $(S-1)$-pure unique pairing condition. If there exists a pair $\{\lambda_i,\lambda_{i+2S-1}\}$ such that $\lambda_i \lambda_{i+2S-1} >1$, we denote it by $\{ \chi_1,\chi_2\}$, where $\chi_1 > \chi_2$, $\chi_1>1$. If $\lambda_i \lambda_{i+2S-1} =1$ for all $i$, then we denote by $\chi_1$ the smallest eigenvalue that's greater than $1$, namely, $\chi_1 = \lambda_{S}$, and we denote $\chi_2 = \chi_1^{-1}$. In both cases the remaining pairs are denoted by $\{\omega_i,\omega_i^{-1} \} $, where $1 < \omega_2 < \hdots < \omega_S $.
\end{definition}
Note that $\omega_i \neq 1$ for any $i$ since that would imply $\omega_i^{-1} = 1$, which breaks the non-degeneracy of $\Lambda$.

It is instructive to provide an example of the set of graphs that are possible when $\Lambda$ satisfies the $(S-1)$-pure unique pairing condition for a particular value of $S$. We consider $S=4$, and we make repeated use of this example to help with conveying the arguments. 
\begin{example}
Consider a non-degenerate $\Lambda = (\lambda_i)_{i=1}^{8}$ of a quantum CM living in $4$ modes. Assume $\Lambda$ satisfies the $3$-pure unique pairing condition. Then, $G(\Lambda)$ belongs to one of the $5$ graphs depicted in Fig.~\ref{fig_S=4_graphs}. A solid segment corresponds to a pure edge, and two vertices $\lambda_i$, $\lambda_j$ share a dashed edge if $\lambda_i \lambda_j > 1$. The absence of an edge between an
 two vertices $\lambda_i$, $\lambda_j$ means $\lambda_i \lambda_j < 1$.

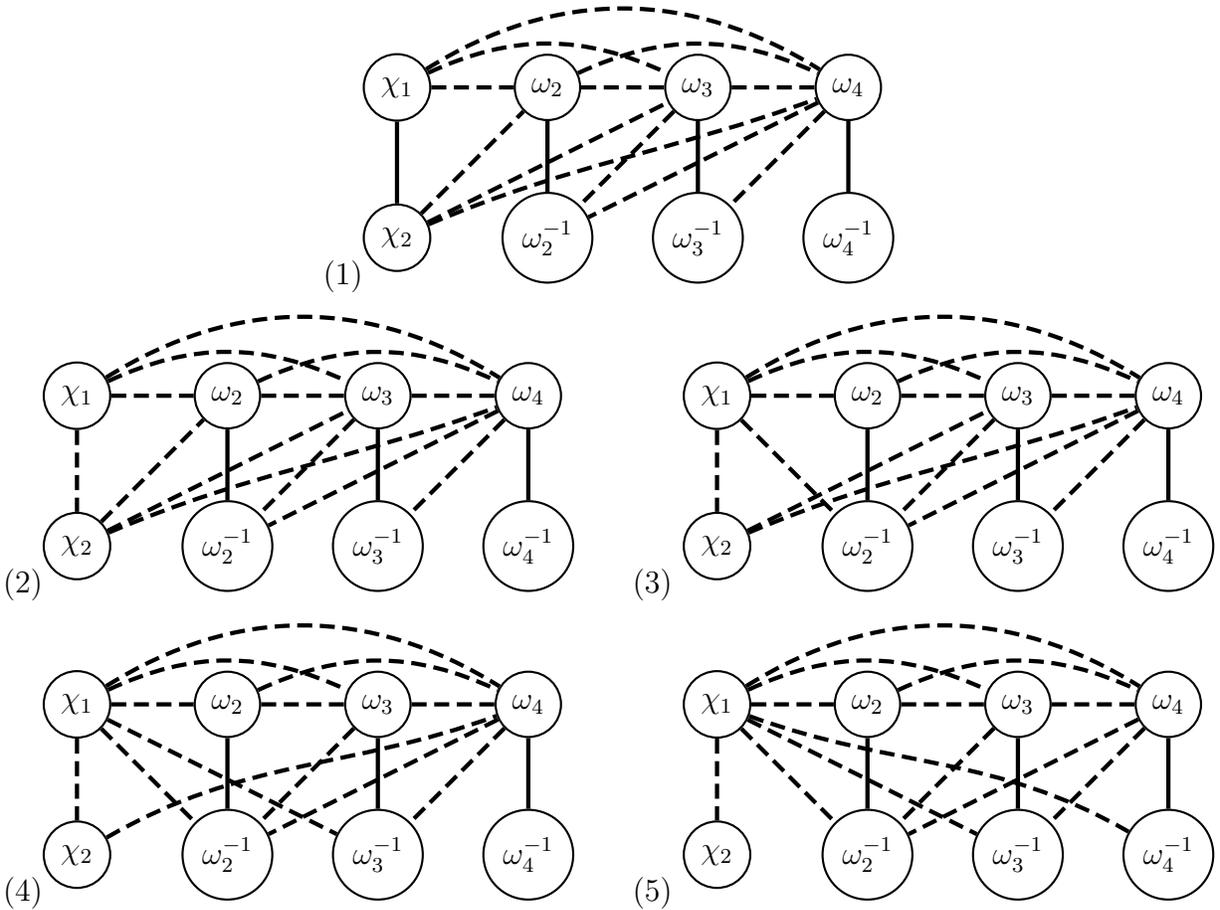
\begin{figure}[t!]

 (1)\begin{tikzpicture}[node distance={20mm},  thick, main/.style = {draw, circle}]
\node[main] (1) {$\chi_1$}; 
\node[main] (2) [right of=1] {$\omega_2$};
\node[main] (3) [ right of=2] {$\omega_3$}; 
\node[main] (4) [right of=3] {$\omega_4$};
\node[main] (5) [below of=1] {$\chi_2$}; 
\node[main] (6) [below of=2] {$\omega_2^{-1}$};
\node[main] (7) [below of=3] {$\omega_3^{-1}$};
\node[main] (8) [below of=4] {$\omega_4^{-1}$};
\draw [dash pattern=on 2mm off 1mm, ultra thick] (1) -- (2);
\draw [dash pattern=on 2mm off 1mm, ultra thick] (1) to[out=25,in=155] (3);
\draw [dash pattern=on 2mm off 1mm, ultra thick] (1) to[out=32,in=148] (4);
\draw [dash pattern=on 2mm off 1mm, ultra thick] (2) -- (3);
\draw [dash pattern=on 2mm off 1mm, ultra thick] (2) to[out=25,in=155] (4);
\draw [dash pattern=on 2mm off 1mm, ultra thick] (3) -- (4);
\draw [ ultra thick] (1) -- (5);
\draw [ ultra thick] (2) -- (6);
\draw [ ultra thick] (3) -- (7);
\draw [ ultra thick] (4) -- (8);
\draw [dash pattern=on 2mm off 1mm, ultra thick] (5) to[out=23,in=200] (4);
\draw [dash pattern=on 2mm off 1mm, ultra thick] (5) -- (3);
\draw [dash pattern=on 2mm off 1mm, ultra thick] (5) -- (2);

\draw [dash pattern=on 2mm off 1mm, ultra thick] (4) -- (6);
\draw [dash pattern=on 2mm off 1mm, ultra thick] (4) -- (7);
\draw [dash pattern=on 2mm off 1mm, ultra thick] (3) -- (6);

\end{tikzpicture}

 (2)\begin{tikzpicture}[node distance={20mm},  thick, main/.style = {draw, circle}]
\node[main] (1) {$\chi_1$}; 
\node[main] (2) [right of=1] {$\omega_2$};
\node[main] (3) [ right of=2] {$\omega_3$}; 
\node[main] (4) [right of=3] {$\omega_4$};
\node[main] (5) [below of=1] {$\chi_2$}; 
\node[main] (6) [below of=2] {$\omega_2^{-1}$};
\node[main] (7) [below of=3] {$\omega_3^{-1}$};
\node[main] (8) [below of=4] {$\omega_4^{-1}$};
\draw [dash pattern=on 2mm off 1mm, ultra thick] (1) -- (2);
\draw [dash pattern=on 2mm off 1mm, ultra thick] (1) to[out=25,in=155] (3);
\draw [dash pattern=on 2mm off 1mm, ultra thick] (1) to[out=32,in=148] (4);
\draw [dash pattern=on 2mm off 1mm, ultra thick] (2) -- (3);
\draw [dash pattern=on 2mm off 1mm, ultra thick] (2) to[out=25,in=155] (4);
\draw [dash pattern=on 2mm off 1mm, ultra thick] (3) -- (4);
\draw [dash pattern=on 2mm off 1mm, ultra thick] (1) -- (5);
\draw [ ultra thick] (2) -- (6);
\draw [ ultra thick] (3) -- (7);
\draw [ ultra thick] (4) -- (8);
\draw [dash pattern=on 2mm off 1mm, ultra thick] (5) to[out=23,in=200] (4);
\draw [dash pattern=on 2mm off 1mm, ultra thick] (5) -- (3);
\draw [dash pattern=on 2mm off 1mm, ultra thick] (5) -- (2);

\draw [dash pattern=on 2mm off 1mm, ultra thick] (4) -- (6);
\draw [dash pattern=on 2mm off 1mm, ultra thick] (4) -- (7);
\draw [dash pattern=on 2mm off 1mm, ultra thick] (3) -- (6);

\end{tikzpicture}  \hspace {0.5cm} (3)  \begin{tikzpicture}[node distance={20mm},  thick, main/.style = {draw, circle}]
\node[main] (1) {$\chi_1$}; 
\node[main] (2) [right of=1] {$\omega_2$};
\node[main] (3) [ right of=2] {$\omega_3$}; 
\node[main] (4) [right of=3] {$\omega_4$};
\node[main] (5) [below of=1] {$\chi_2$}; 
\node[main] (6) [below of=2] {$\omega_2^{-1}$};
\node[main] (7) [below of=3] {$\omega_3^{-1}$};
\node[main] (8) [below of=4] {$\omega_4^{-1}$};
\draw [dash pattern=on 2mm off 1mm, ultra thick] (1) -- (2);
\draw [dash pattern=on 2mm off 1mm, ultra thick] (1) to[out=25,in=155] (3);
\draw [dash pattern=on 2mm off 1mm, ultra thick] (1) to[out=32,in=148] (4);
\draw [dash pattern=on 2mm off 1mm, ultra thick] (2) -- (3);
\draw [dash pattern=on 2mm off 1mm, ultra thick] (2) to[out=25,in=155] (4);
\draw [dash pattern=on 2mm off 1mm, ultra thick] (3) -- (4);
\draw [dash pattern=on 2mm off 1mm, ultra thick] (1) -- (5);
\draw [ ultra thick] (2) -- (6);
\draw [ ultra thick] (3) -- (7);
\draw [ ultra thick] (4) -- (8);
\draw [dash pattern=on 2mm off 1mm, ultra thick] (5) to[out=23,in=200] (4);
\draw [dash pattern=on 2mm off 1mm, ultra thick] (5) -- (3);
\draw [dash pattern=on 2mm off 1mm, ultra thick] (1) -- (6);

\draw [dash pattern=on 2mm off 1mm, ultra thick] (4) -- (6);
\draw [dash pattern=on 2mm off 1mm, ultra thick] (4) -- (7);
\draw [dash pattern=on 2mm off 1mm, ultra thick] (3) -- (6);

\end{tikzpicture}

(4)\begin{tikzpicture}[node distance={20mm},  thick, main/.style = {draw, circle}]
\node[main] (1) {$\chi_1$}; 
\node[main] (2) [right of=1] {$\omega_2$};
\node[main] (3) [ right of=2] {$\omega_3$}; 
\node[main] (4) [right of=3] {$\omega_4$};
\node[main] (5) [below of=1] {$\chi_2$}; 
\node[main] (6) [below of=2] {$\omega_2^{-1}$};
\node[main] (7) [below of=3] {$\omega_3^{-1}$};
\node[main] (8) [below of=4] {$\omega_4^{-1}$};
\draw [dash pattern=on 2mm off 1mm, ultra thick] (1) -- (2);
\draw [dash pattern=on 2mm off 1mm, ultra thick] (1) to[out=25,in=155] (3);
\draw [dash pattern=on 2mm off 1mm, ultra thick] (1) to[out=32,in=148] (4);
\draw [dash pattern=on 2mm off 1mm, ultra thick] (2) -- (3);
\draw [dash pattern=on 2mm off 1mm, ultra thick] (2) to[out=25,in=155] (4);
\draw [dash pattern=on 2mm off 1mm, ultra thick] (3) -- (4);
\draw [dash pattern=on 2mm off 1mm, ultra thick] (1) -- (5);
\draw [ ultra thick] (2) -- (6);
\draw [ ultra thick] (3) -- (7);
\draw [ ultra thick] (4) -- (8);
\draw [dash pattern=on 2mm off 1mm, ultra thick] (5) to[out=30,in=200] (4);
\draw [dash pattern=on 2mm off 1mm, ultra thick] (1) -- (6);
\draw [dash pattern=on 2mm off 1mm, ultra thick] (1) -- (7);

\draw [dash pattern=on 2mm off 1mm, ultra thick] (4) -- (6);
\draw [dash pattern=on 2mm off 1mm, ultra thick] (4) -- (7);
\draw [dash pattern=on 2mm off 1mm, ultra thick] (3) -- (6);

\end{tikzpicture} \hspace{0.5 cm} (5)
\begin{tikzpicture}[node distance = 20mm,  thick, main/.style = {draw, circle}]
\node[main] (1) {$\chi_1$}; 
\node[main] (2) [right of = 1] {$\omega_2$};
\node[main] (3) [ right of=2] {$\omega_3$}; 
\node[main] (4) [right of=3] {$\omega_4$};
\node[main] (5) [below of=1] {$\chi_2$}; 
\node[main] (6) [below of=2] {$\omega_2^{-1}$};
\node[main] (7) [below of=3] {$\omega_3^{-1}$};
\node[main] (8) [below of=4] {$\omega_4^{-1}$};
\draw [dash pattern=on 2mm off 1mm, ultra thick] (1) -- (2);
\draw [dash pattern=on 2mm off 1mm, ultra thick] (1) to[out=25,in=155] (3);
\draw [dash pattern=on 2mm off 1mm, ultra thick] (1) to[out=32,in=148] (4);
\draw [dash pattern=on 2mm off 1mm, ultra thick] (2) -- (3);
\draw [dash pattern=on 2mm off 1mm, ultra thick] (2) to[out=25,in=155] (4);
\draw [dash pattern=on 2mm off 1mm, ultra thick] (3) -- (4);
\draw [dash pattern=on 2mm off 1mm, ultra thick] (1) -- (5);
\draw [ultra thick] (2) -- (6);
\draw [ultra thick] (3) -- (7);
\draw [ ultra thick] (4) -- (8);
\draw [ dash pattern=on 2mm off 1mm, ultra thick] (1) -- (6);
\draw [ dash pattern=on 2mm off 1mm, ultra thick] (1) -- (7);
\draw [ dash pattern=on 2mm off 1mm, ultra thick] (1) to[out=340,in=150] (8);

\draw [dash pattern=on 2mm off 1mm, ultra thick] (4) -- (6);
\draw [dash pattern=on 2mm off 1mm, ultra thick] (4) -- (7);
\draw [dash pattern=on 2mm off 1mm, ultra thick] (3) -- (6);

\end{tikzpicture}

\caption{The set of graphs $G(\Lambda)$ that are possible for non-degenerate $\Lambda$ satisfying the $(S-1)$-pure unique pairing condition when $S=4$. The vertices are labeled by the eigenvalues. The solid and dashed edges connect vertices associated with eigenvalues whose product is $=1$ and $>1$, respectively. If two vertices $\lambda_i$ and $\lambda_j$ do not share an edge, then $\lambda_i\lambda_j <1$.}
\label{fig_S=4_graphs}
\end{figure}

\end{example}
Let us dissect the example. The edges connecting the pairs $\{\omega_i,\omega_i^{-1}\}$ are solid since $\omega_i \omega_i^{-1}=1$. Notice that in graph (1) $\{\chi_1, \chi_2 \}$ is solid, while in the rest of the graphs it is dashed. There cannot exist two solid edges connected to any single vertex as that would imply the two vertices at the other ends of the edges are equal, thus breaking the non-degeneracy of $\Lambda$. For each graph we first note that since $\omega_4 > \omega_3 > \omega_2$ in the top row of vertices, then $\omega^{-1}_4 < \omega^{-1}_3 < \omega^{-1}_2$ in the bottom row. This is why $\omega_4$ is connected to both $\omega_2^{-1}$ and to $\omega_3^{-1}$ with dashed lines, and similarly for the edge $\{ \omega_3,\omega_2^{-1} \}$, as well as why there are no edges between $\omega_i$ and $\omega_j^{-1}$ for $j>i$. This logic extends to vertices $\chi_1,\chi_2$ of graph (1), where $ \chi_1\chi_2=1$ and $\chi_1 < \omega_2 < \hdots < \omega_S $. All vertices in the top rows are connected by dashed edges since $\chi_1>1$ and $\omega_i>1$ for all $i$. The next feature to notice is that $\chi_2$ is not connected to any of the $\omega_i^{-1}$. This is because otherwise the perfect matching would not be unique, namely, one could also pair $\chi_2$ with $\omega_i^{-1}$ and $\chi_1$ with $\omega_i$. 

We can thus see that the main differences between the graphs $(2)-(5)$ lie in the connections between $\chi_1$ and the $\omega_i^{-1}$, and in the connections between $\chi_2$ and the $\omega_i$. To explain these connections, we first note that the edges $\{\chi_1,\omega_i^{-1}  \}$ and $\{\chi_2,\omega_i  \}$ cannot exist simultaneously for any $i$ as that would introduce another distinct perfect matching. Next, because of the ordering of the $\omega_i$, if the edge $\{\chi_1,\omega_i^{-1}  \}$ exists, then so do the edges $\{\chi_1,\omega_j^{-1}  \}$ for all $j<i$. Similarly, if the edge $\{\chi_2,\omega_i  \}$ exists, then so do the edges $\{\chi_2,\omega_j  \}$ for all $j>i$. The final feature to note is that, for each $i$, either $\{\chi_1,\omega_i^{-1}  \}$ or $\{\chi_2,\omega_i  \}$ exists. This is because if neither edge exists, then one simultaneously has $\chi_1 \omega_i^{-1}  <1$ and $\chi_2 \omega_i <1$, and that implies $\chi_1 \chi_2 <1$. Thus, for example, one way to specify one of the graphs where $ \chi_1\chi_2>1$ is to count the number of edges connecting $\chi_1$ with the $\omega_i^{-1}$. Since there are $4$ possibilities, including the case with no edges, there are $4$ distinct graphs.

The observations in the previous paragraph generalize to arbitrary $S$. More specifically, the number of distinct possible graphs equals $S+1$, where $ \{\chi_1, \chi_2 \}$ is pure in only one of them. The remaining graphs can be distinguished by the highest value of $i$ such that $\{\chi_1,\omega_i^{-1}\}$ exists, including the case where there is no such value (e.g. graph (2) in Fig.~\ref{fig_S=4_graphs}). 

\begin{definition}\label{define class_of_spectra} 
We denote the class of non-degenerate eigenspectra that satisfy the $(S-1)$-pure unique pairing condition by $\cC(S)$. We partition $\cC(S)$ into $S+1$ subclasses, $\{\cC_i(S)\}_{i=0}^{S}$. The subclass $\cC_0(S)$ is composed of all $\Lambda$ that have a graph $G(\Lambda)$ where $\{\chi_1,\chi_2\}$ is pure. For each $1 < i < S$, the subclass $\cC_i(S)$ is composed of all $\Lambda$ that have a graph $G(\Lambda)$ where $\chi_1$ has an edge with each $\omega_j^{-1}$ where $j \leq i$, and no edge when $j>i$. $\cC_1(S)$ ($\cC_S(S)$) is composed of all $\Lambda$ that have a graph $G(\Lambda)$ where $\chi_1$ has no edge with any $\omega_j^{-1}$ ($\chi_1$ has an edge with every $\omega_j^{-1}$). 
\end{definition}

 \section{Classification of eigenspectra}\label{sec main_result}

In this section we present our main findings. We start by deriving certain expressions for the $\det (\Gamma^{(i)})$ (see the end of Sec.~\ref{sec prob_form} for the definition) of an arbitrary real symmetric matrix $\Gamma$. 
\begin{lemma}\label{lem det(Gamma_i) }
Consider a diagonal matrix $D =  \mathrm{diag} [d_1,\hdots,d_{2S}]$, and an arbitrary $O \in \mathrm{O}(2S)$. Let $\Gamma = O^T D O$. Then, 
\begin{equation}
\det (\Gamma^{(i)}) = \sum_{k>l}  \left[\det    ( O_{lk}^{2i-1,2i})   \right]^2 d_{l} d_{k},
\end{equation}
for each $i = 1,\hdots,S$.
\end{lemma}
\begin{proof}
We prove by direct derivation. For a given $i$, let us write $\Gamma^{(i)}$ explicitly, that is,
\begin{align}\label{eq: D_1}
(\Gamma^{(i)})_{nm} & = \sum_{k=1}^{2S} \sum_{l=1}^{2S}  (O^T)_{n,k}D_{kl}O_{lm},
\end{align}
for $n = 2i-1,2i$ and $m=2i-1,2i$. Notice that $\Gamma^{(i)}$ can be written as $\Gamma^{(i)} = (\tilde O)^T D \tilde O$, where $\tilde O = (c_{2i-1}[O], c_{2i}[O]  )$. We can thus use the Cauchy-Binet formula for the product of a $2 \times (2S-2)$ matrix, $(\tilde O)^T$, with a $(2S-2) \times 2$ matrix, $D \tilde O$:
\begin{align}
\det(\Gamma^{(i)}) &= \sum_{k>l} \det   \begin{pmatrix} \tilde O^T_{1,l} & \tilde O^T_{1,k} \\ \tilde O^T_{2,l} & \tilde O^T_{2,k}  \end{pmatrix} \det  \begin{pmatrix} (D \tilde O)_{l,1} &(D \tilde O)_{l,2} \\ (D \tilde O)_{k,1} & (D \tilde O)_{k,2}  \end{pmatrix} \nonumber \\
& =  \sum_{k>l} \det   \begin{pmatrix} \tilde O_{l,1} & \tilde O_{l,2} \\ \tilde O_{k,1} & \tilde O_{k,2}  \end{pmatrix} \det  \begin{pmatrix} (D \tilde O)_{l,1} &(D \tilde O)_{l,2} \\ (D \tilde O)_{k,1} & (D \tilde O)_{k,2}  \end{pmatrix},
\end{align}
where the sum is over all $k,l$ in $2S \geq k >l \geq 1$.

We use the fact that $D$ is diagonal to express $\begin{pmatrix} (D \tilde O)_{l,1} &(D \tilde O)_{l,2} \\ (D \tilde O)_{k,1} & (D \tilde O)_{k,2}  \end{pmatrix} = \begin{pmatrix} d_l \tilde O_{l,1} &   d_l \tilde O_{l,2} \\ d_k \tilde O_{k,1} & d_k \tilde O_{k,2}  \end{pmatrix}$. Then,
\begin{align}
\det ( \Gamma^{(i)} ) &= \sum_{k>l} \det    \begin{pmatrix} \tilde O_{l,1} & \tilde O_{l,2} \\ \tilde O_{k,1} & \tilde O_{k,2}  \end{pmatrix}  \det   \begin{pmatrix} d_l \tilde O_{l,1} &   d_l \tilde O_{l,2} \\ d_k \tilde O_{k,1} & d_k \tilde O_{k,2}  \end{pmatrix} \nonumber \\
& = \sum_{k>l}  \left( \det     \begin{pmatrix} \tilde O_{l,1} & \tilde O_{l,2} \\ \tilde O_{k,1} & \tilde O_{k,2}  \end{pmatrix}  \right)^2 d_{l} d_{k} \nonumber \\
& = \sum_{k>l}  \left( \det    \begin{pmatrix}   O_{l,2i-1} &  O_{l,2i} \\  O_{k,2i-1} &  O_{k,2i}  \end{pmatrix}  \right)^2 d_{l} d_{k} .
\end{align}
   
\end{proof}
According to Fact~\ref{O_det_conv}, the set $\big\{ \left[\det    ( O_{lk}^{2i-1,2i})   \right]^2  \big\}_{k>l}  $ in Lem~\ref{lem det(Gamma_i) } is a set of non-negative weights that sum to $1$. Hence, each $\det(\Gamma^{(i)})$ is a convex combination of the set of products $\{d_ld_k\}_{k>l} $. Thus, if $D$ is a quantum CM, one way to obtain $\det(\Gamma^{(i)})<1$ is if $\det    ( O_{lk}^{2i-1,2i}) =0$ for all $\{l,k\}$ where $d_ld_k >1$, and there exists a pair $ \{ n,m \}$ such that $\det    ( O_{nm}^{2i-1,2i}) \neq 0$ and $d_n d_m <1$. This observation is central to the strategy used in the proof of Thm.~\ref{thm sub}, in constructing the "uncertainty-principle-violating witness" alluded to earlier.

\begin{theorem}\label{thm sub}
Let $\Gamma$ and $\Gamma'$ be two quantum CMs with the same non-degenerate eigenspectrum $\Lambda= (\lambda_i)_{i=1}^{2S}$.
Suppose that $\Lambda$ satisfies the $(S-1)$-pure unique pairing
condition. Then $\Gamma' = K^T\Gamma K$ for some $K \in \mathrm{SpO}(2S,\mathbb{R})$.
\end{theorem}
\begin{proof}

We want to show that given any quantum CM $\Gamma$ with the eigenspectrum $\Lambda$, and any ONS matrix $O$, $O^T \Gamma O$ cannot be a quantum CM. It suffices to show that there exists some $D \in \cD_\Lambda$, such that $O^T D O$ is not a quantum CM for every ONS matrix $O$. This is because, for any $\Gamma$ with the eigenspectrum $\Lambda$, $\Gamma = (O')^T D O'$ for some $O' \in \mathrm{O}(2S)$. Thus, if we show the former, that would imply $O' \in \mathrm{SpO}(2S,\mathbb{R})$, which, in turn, implies the theorem statement.

Now, if $O^T D O$ is not a quantum CM, then the same is true for any $W^TO^T D O W$, where $W \in \mathrm{SpO}(2S,\mathbb{R})$, and vice versa. This is because symplectic transformations do not affect the symplectic eigenvalues. Our proof strategy is to show that there exists $W \in  \mathrm{SpO}(2S,\mathbb{R})$ such that $\det ( (W^TO^T D O W)^{(i)})<1  $ for some $1 \leq i \leq S$. 

Recall Def.~\ref{define class_of_spectra}. The proof steps are the same for the classes $\cC_0(S),\hdots,\cC_{S-1}(S)$, while a slightly different method is used for the class $\cC_S(S)$. We treat the former first. 

\begin{casse}
The proof for $\Lambda \in \bigcup\limits_{i=0}^{S-1} \cC_i(S) $.
\end{casse}

Consider a maximally mode-reduced transform of an ONS matrix $O$, $OW_1 = I_{2S-2n} \oplus O'$, for some $W_1 \in \mathrm{SpO}(2S,\mathbb{R})$ and $O' \in   \mathrm{O}(2n) $. Recall that the maximally reduced mode number $n$ satisfies $S \geq n \geq 2$ according to Lem.~\ref{lem n_geq_2}. Let us denote $i_n = S-n+1$. Then, notice that $\det ( (OW_1)_{kl}^{2i_n-1,2i_n} ) = 0$ for all $k<l$ where $ k \leq 2S-2n$ or $ l \leq 2S-2n$. According to Lem.~\ref{lem col12_O_rep}, there exists $W_2 \in \mathrm{SpO}(2n,\mathbb{R})$ such that $c_{1}[O'W_2]_k = \lambda c_{2}[O'W_2]_k$ for $k=2,3,\hdots,2n$, for some constant $\lambda \neq 0$. Then, $W = W_1 ( I_{2S-2n} \oplus W_2)$ is the matrix we are looking for. We proceed to show this.

First, note that $OW = OW_1 ( I_{2S-2n} \oplus W_2) = (I_{2S-2n} \oplus O') ( I_{2S-2n} \oplus W_2) =  I_{2S-2n} \oplus O'W_2$. Thus, $\det ( (OW)_{kl}^{2i_n-1,2i_n} ) = 0$ for all $k<l$ where $ k \neq 2i_n-1$ or $ l \leq 2S-2n$. Next, consider the following $D =\textrm{diag}[d_1,\hdots,d_{2S}]  \in \cD_\Lambda$:
\begin{equation}
D = \textrm{diag}[\omega^{-1}_{S-n},\omega_{S-n},\omega^{-1}_{S-n+1},\omega_{S-n+1},\hdots,\omega^{-1}_{S-1},\omega_{S-1}  ,\omega^{-1}_S,\omega_S,  \hdots   ,\chi_2,\chi_1].
\end{equation}
Here, the second ``$\hdots$" denotes some valid arrangement of the $\omega_i,\omega_i^{-1} $ for $i = 2,\hdots, S-n-1$. Importantly, notice that the eigenvalue $\omega_{S}^{-1}$ occupies the $(2i_n-1)$'th entry of the main diagonal of $D$.

Using Lem.~\ref{lem det(Gamma_i) },
\begin{align}
\det ( (  W^TO^T D O W )^{(i_n)}  ) & =  \sum_{k<l}  \left[\det    ( (OW)_{kl}^{2i_n-1,2i_n})   \right]^2 d_{l} d_{k} \nonumber \\
& =  \sum_{l=2i_n}^{2S}  \left[\det    ( (OW)_{2i_n-1,l}^{2i_n-1,2i_n})   \right]^2 d_{l} d_{2i_n-1} \nonumber \\
& =  \sum_{l=2i_n}^{2S}  \left[\det    ( (OW)_{2i_n-1,l}^{2i_n-1,2i_n})   \right]^2 d_{l} \omega^{-1}_S.
\end{align}
Thus, $\det ( (  W^TO^T D O W )^{(i_n)}  )$ is a convex combination of the set of products \[ \{ \omega_S^{-1}\omega_S  \} \cup \{ \omega_S^{-1}\omega_j \}_{j=2}^{S-n-1} \cup \{ \omega_S^{-1}\omega^{-1}_j \}_{j=2}^{S-n-1} \cup \{ \omega_S^{-1}\chi_j\ \}_{j=1,2} .  \] These products are all less than $1$ except for $  \omega_S^{-1}\omega_S = 1 $. For pictorial illustration, consider the graphs (1)-(4) in Fig.~\ref{fig_S=4_graphs} for $S=4$, where the vertex $\omega_4^{-1}$ has an edge with only the vertex $\omega_4$. Thus, $\det ( (  W^TO^T D O W )^{(i_n)}  )  \leq 1$, with $ \det ( (  W^TO^T D O W )^{(i_n)}  )  =1$ if and only if $\det ( (OW)_{2i_n-1,l}^{2i_n-1,2i_n}) = 0 $ for all $l \neq 2i_n$. The latter is not possible since it entails that $\det    ( (OW)_{2i_n-1,2i_n}^{2i_n-1,2i_n}) = 1 $, which is in violation of Lem.~\ref{lem det=1_impos}. More specifically, according to Lem.~\ref{lem det=1_impos} there does not exist $ W_2 \in  \mathrm{SpO}(2n,\mathbb{R})$ such that $\det    ( (  OW_1 ( I_{2S-2n} \oplus W_2) )_{2i_n-1,2i_n}^{2i_n-1,2i_n}) = 1 $ when $OW_1 $ is a maximally mode-reduced transform of $O$ with the maximally reduced mode number $n$. Thus, $ \det ( (  W^TO^T D O W )^{(i_n)}  )  <1$.

\begin{casse}
The proof for $\Lambda \in \cC_S(S) $.
\end{casse}
Again, consider a maximally mode-reduced transform of an ONS matrix $O$, $OW_1 = I_{2S-2n} \oplus O'$, for some $W_1 \in \mathrm{SpO}(2S,\mathbb{R})$ and $O' \in   \mathrm{O}(2n) $. Again, denote $i_n = S-n+1$. According to Lem.~\ref{lem col12_O_rep_3} there exists $W_2 \in \mathrm{SpO}(2n,\mathbb{R})$ such that the rows $ 4,6,\hdots, 2n$ in $c_{1}[O'W_2]$ and $c_{2}[O'W_2]$ vanish. Then, $W = W_1 ( I_{2S-2n} \oplus W_2)$ is the matrix we are looking for, which we now proceed to show.

First, note that $OW = I_{2S-2n} \oplus O'W_2$. Thus, $\det ( (OW)_{kl}^{2i_n-1,2i_n} ) \neq 0$ only when $k,l \in \{2i_n-1,2i_n\} \cup \{ 2i_n+1, 2i_n+3,\hdots,2S-1  \}$. Now, consider the following $D =\textrm{diag}[d_1,\hdots,d_{2S}]  \in \cD_\Lambda$:
\begin{equation}
D = \textrm{diag}[\hdots, \omega^{-1}_2,\omega_2, \omega^{-1}_3,\omega_3, \hdots,  \omega^{-1}_{ n },\omega_{n  },\chi_2,\chi_1].
\end{equation}
Here, the first ``$\hdots$" denotes some valid arrangement of the $\omega_i,\omega_i^{-1} $ for $i = n+1,\hdots, S$. Notice that the eigenvalue $\omega_{2}^{-1}$ occupies the $(2i_n-1)$'th entry of the main diagonal of $D$. Next, using Lem.~\ref{lem det(Gamma_i) },
\begin{align}
& \det ( (  W^TO^T D O W )^{(i_n)}  )  =  \sum_{k<l}  \left[\det     (OW)_{kl}^{2i_n-1,2i_n}   \right]^2 d_{l} d_{k} \nonumber \\
& = \left[\det     (OW)_{2i_n-1,2i_n}^{2i_n-1,2i_n}   \right]^2 d_{2i_n-1}d_{2i_n}  +  \sum\limits_{\substack{k = 2i_n-1 \\ k \textrm{ odd}}}^{2S-3} \sum_{\substack{l = k+2 \\ l \textrm{ odd}}}^{2S-1} \left[\det     (OW)_{kl}^{2i_n-1,2i_n} \right]^2 d_{l} d_{k}  + \hdots \nonumber \\
 & + \sum_{\substack{l = 2i_n+1 \\ l \textrm{ odd}}}^{2S-1} \left[\det     (OW)_{2i_n,l}^{2i_n-1,2i_n} \right]^2 d_{2i_n} d_{l}.
\end{align}
Thus, $\det ( (  W^TO^T D O W )^{(i_n)}  )$ is a convex combination of the products in the set
\begin{align}
\{ \omega_j^{-1}\omega_2  \}_{j=2}^n \cup \{ \chi_2 \omega_2 \} \cup \{ \{ \omega_j^{-1}\omega^{-1}_k \}_{k=j+1}^n \}_{j=2}^{n-1}  \cup \{ \omega_j^{-1}\chi_2 \}_{j=2}^n .
\end{align} 
Each of the products, with the exception of $\omega_2^{-1}\omega_2 = 1$, is less than $1$. For pictorial illustration, consider graph (5) of Fig.~\ref{fig_S=4_graphs} for $S=4$, where the only edge among the pairs of vertices appearing in the set of products above is between $\omega_2^{-1}$ and $\omega_2$. Thus, $ \det ( (  W^TO^T D O W )^{(i_n)}  )  <1$ unless $\det ( (OW)_{2i_n-1,2i_n}^{2i_n-1,2i_n}) = 1 $. But the latter is impossible according to Lem.~\ref{lem det=1_impos}.
\end{proof}

 It follows from Thm.~\ref{thm sub} that, given $\Lambda \in \cC(S)$, all corresponding quantum CMs have the same thermal and squeezing parameters. This is a direct consequence of the fact that the set of quantum CMs form a single orbit under the action of $\mathrm{SpO}(2S,\mathbb{R})$. Thus, $ \cC(S) \subset \cP_1(S)$ and $ \cC(S) \subset \cP_2(S)$. We can find examples of $\Lambda \notin \cC(S)$ such that $\Lambda \in \cP_1(S)$ and $\Lambda \in \cP_2(S)$. For any $S$, any $\Lambda$ where all eigenvalues are equal clearly falls in this category. For such eigenspectra, the orbit under the action of $\mathrm{SpO}(2S,\mathbb{R})$ is composed of a single CM. But we can find a more general class of examples.
\begin{lemma}\label{lemma 2S-1-deg eigenspectrum}
Consider an eigenspectrum $\Lambda =  (\lambda_i )_{i=1}^{2S}$ of a quantum CM, where $\lambda_2 = \hdots = \lambda_{2S}$. Then, $\Lambda \in \cP_1$, $\Lambda \in \cP_2$. 
\end{lemma}
\begin{proof}
Consider $D = \textrm{diag} [\lambda_1,\hdots, \lambda_{2S} ] \in \cD_\Lambda$. Then, $O D O^T$ is a trivial ONS transformation for every ONS matrix $O$. In particular, by applying Lem.~\ref{lem OW=1_oplus_O'}, we can write $O =  O' W $, where $O' = I_1 \oplus O''$ for some $O'' \in \mathrm{O}(2S-1) $, and $W \in \mathrm{SpO}(2S,\mathbb{R})$. And it can be seen that $(O')^TDO'  = D$. Hence, every quantum CM can be written as $WDW^T$ for some $W \in \mathrm{SpO}(2S,\mathbb{R})$, and, therefore, has the same set of thermal and squeezing parameters as $D$.
 \end{proof}

Notice that the set of examples of eigenspectra we described above have degeneracies. Further, they satisfy a more general version of the unique pairing condition.
\begin{definition}
Consider an eigenspectrum $\Lambda = (\lambda_i )_{i=1}^{2S}$ of a quantum CM, where the $\lambda_i$ are in a non-ascending order. Recall that the pairing $\{ \{ \lambda_i,\lambda_{2S-i+1}  \}    \}_{i=1}^S$ satisfies $\lambda_i\lambda_{2S-i+1} \geq 1$ for all $i$. We say $\Lambda$ satisfies the general unique pairing condition if all valid pairings of the $\lambda_i$ are identical to $\{ \{ \lambda_i,\lambda_{2S-i+1}  \}    \}_{i=1}^S$. 
\end{definition}
Note that the general unique pairing condition is equivalent to the unique pairing condition when the eigenspectrum is non-degenerate.
\begin{remark}\label{remark Lambda_not_uniqpair}
Consider an eigenspectrum $\Lambda$ of a quantum CM. Assume $\Lambda$ does not satisfy the general unique pairing condition. Then, $\Lambda \notin \cP_1$, $\Lambda \notin \cP_2$.
\end{remark}
At the risk of seeming over-explanatory, let us comment on this remark. First, if $\Lambda$ does not satisfy the general unique pairing condition, then there are at least two sets of thermal and squeezing parameters compatible with $\Lambda$. This is because the thermal and squeezing parameters corresponding to a given pairing of eigenvalues determine and are determined by the pairing, as can be seen from Fact~\ref{remark u_and_r_det_diag}. The set of quantum CMs with such a $\Lambda$ cannot be in a single orbit under the action of $\mathrm{SpO}(2S,\mathbb{R})$. This is because, given $D_0 = \textrm{diag}[d_1,\hdots,d_{2S}]$ such that $\{ \{d_{2i-1},d_{2i} \}   \}_{i=1}^S$ is a valid pairing that is different from $\{ \{ \lambda_i,\lambda_{2S-i+1}  \}    \}_{i=1}^S$, and given any $D \in \cD_\Lambda$, there exists $P \in \textrm{Sym}(S)\setminus \mathrm{SpO}(2S,\mathbb{R})$ such that $D = PD_0 P^T$.

Thus, $\cP_1(S)$ and $\cP_2(S)$ must necessarily each be a subset of the class of eigenspectra that satisfy the general unique pairing condition, but be strictly larger than the class of non-degenerate eigenspectra that satisfy the $(S-1)$-pure unique pairing condition. Our last set of results derived in this paper aim to narrow down the set of candidate eigenspectra that might belong to $\cP_1(S)$ and/or $\cP_2(S)$. In the following proposition statement, a tuple is ``at most $2$-fold degenerate" if there is no subset of three elements that are equal.

\begin{prop}\label{prop}
Consider the eigenspectrum $\Lambda =  \{ \nu_i e^{-2r_i},\nu_i e^{2r_i} \}_{i=1}^S$ of a quantum CM, where the eigenvalues are explicitly associated with thermal and squeezing parameters. 
Assume $\Lambda$ satisfies the general unique pairing condition. Further, assume there exists a $4$-tuple $\{\nu_k e^{-2r_k},\nu_k^{2r_k},\nu_l e^{-2r_l},\nu_l e^{2r_l}\} \subseteq \Lambda$ that has at most $2$-fold degeneracy, and $\nu_k >1$, $\nu_l>1$. Then $\Lambda \notin \cP_1$, $\Lambda \notin \cP_2$.
\end{prop}
\begin{proof}
Since $\cP_1 \subseteq \cP_2$, it suffices to show that $\Lambda \notin \cP_2$. We can assume $k =1$ and $l=2$, and $\nu_1\leq \nu_2$, without loss of generality. Consider the following $D \in \cD_\Lambda$:
\begin{align}
D =  \mathrm{diag} [\nu_1 e^{-2r_1},\nu_1 e^{2r_1},\hdots, \nu_S e^{-2 r_S},\nu_S e^{ 2 r_S}].
\end{align}
We aim to show that there exists a non-trivial ONS transformation of $D$ by some $O \in \mathrm{O}(2S)$, such that $O^T DO$ is a quantum CM and has a set of thermal parameters that is different from $\{ \nu_i \}_{i=1}^S$.

We constrain $O$ to be of the form $O  = O' \oplus I_{2S-4}$, where $O' \in \mathrm{O}(4)$. This implies $O^T DO$ has the form $(O')^T D_4 O' \oplus D_{2S-4}$, where $D_4 = \mathrm{diag} [\nu_1 e^{-2r_1},\nu_1 e^{2r_1},\nu_2 e^{-2 r_2},\nu_2 e^{ 2 r_2}]$ and $D_{2S-4} = \mathrm{diag} [\nu_3 e^{-2r_3},\nu_3 e^{2r_3},\hdots, \nu_S e^{-2 r_S},\nu_S e^{ 2 r_S}]$. Hence, the thermal parameters $\nu_3,\hdots,\nu_S$ are not affected by this transformation. Let us denote the remaining two symplectic eigenvalues of $O^T DO$ by $\nu_1',\nu_2'$. Note that these are the symplectic eigenvalues of $(O')^T D_4 O'$. Thus, we want to show that $O'$ can be chosen such that $\nu_1' \geq 1, \nu_2' \geq 1$, and $\{\nu'_1,\nu'_2  \} \neq \{ \nu_1,\nu_2 \}$.

 We start by writing $D_4$ as $D_4 = \textrm{diag}[\gamma_1,\gamma_2,\gamma_3,\gamma_4]$, so that $\nu_1 = \sqrt{\gamma_1 \gamma_2}$ and $\nu_2 = \sqrt{\gamma_3 \gamma_4}$. Since the $\gamma_i$ are at most $2$-fold degenerate, it follows that there exists a pairing $ \{\{\gamma_{\sigma(1)}, \gamma_{\sigma(2)} \}, \{\gamma_{\sigma(3)}, \gamma_{\sigma(4)} \} \}$, where $\sigma \in \textrm{Sym}(4)$, such that $\sqrt{\gamma_{\sigma(1)} \gamma_{\sigma(2)}} \neq \nu_i$ for $i=1,2$. Since $\Lambda$ satisfies the general unique pairing condition, either $\sqrt{\gamma_{\sigma(1)} \gamma_{\sigma(2)}} <1$ or $\sqrt{\gamma_{\sigma(3)} \gamma_{\sigma(4)}} <1$. We can choose $\sigma$ such that $\sqrt{\gamma_{\sigma(1)} \gamma_{\sigma(2)}} <1$ without loss of generality. Let $P$ denote the permutation matrix corresponding to $\sigma$, that is, $P^T D_4 P = \textrm{diag}[\gamma_{\sigma(1)}, \gamma_{\sigma(2)},\gamma_{\sigma(3)}, \gamma_{\sigma(4)} ] $.

Now, $\nu_1'$ and $\nu_2'$ are given by the absolute values of the eigenvalues of $\Omega_2 (O')^T D_4 O'$, where $\Omega_2 = \Omega_1 \oplus \Omega_1$ \cite{weedbrook:qc2012a}. Consider the map $\tau$ from the cone of $4 \times 4$ PDMs, $\mathbb{P}_{4}(\mathbb{R})$, to the cone $\mathbb{K}$ in $\mathbb{R}^{2}$ consisting of all tuples $(k_1,k_2)$, where $k_2 \geq k_1 > 0$, which takes any $A \in \mathbb{P}_{4}(\mathbb{R})$ to its symplectic eigenvalues arranged in a non-decreasing order. It can be shown that $\tau$ is continuous in the standard topology of both spaces, see, for example, Thm.~7 in \cite{bhatia2015on}. 

Next, note that the subset $\{ O_4^T D_4 O_4 \mid O_4 \in \mathrm{O}(4) \} \subset \mathbb{P}_{4} (\mathbb{R})$ is a path connected region. Thus, for any $O_4 \in \textrm{O}(4) $, there exists a continuous curve $O^T_4(t)D_4 O_4(t) \subset \mathbb{P}_{4} (\mathbb{R}) $ parameterized by $t \in [0,t_f]$, where $O_4(0) = I$ and $O_4 (t_f) = O_4$. This curve produces a continuous curve $( k_1(t),k_2(t) )$ in $\mathbb{K}$ by the continuity of $\tau$. We can choose $O_4 = P$ so that $k_1(t_f) = \sqrt{\gamma_{\sigma(1)} \gamma_{\sigma(2)}} <1$. Since we assumed $k_1(0) = \nu_1>1$ and $k_2(0) = \nu_2>1$, there exists a $t_p \in   [0,t_f]$ such that $k_1(t_p) \geq 1$, $k_2(t_p)\geq 1$ and $k_1(t_p) \neq \nu_i$ for $i =1,2$. Thus, the matrix $O'$ that we are looking for can be taken to be $O' = O_4(t_p)$.

 \end{proof}
It is now easy to show that, when restricted to non-degenerate eigenspectra, we have solved both Prob.~\ref{problem 1} and Prob.~\ref{problem 2}.
 \begin{corollary}\label{cor non-deg-class}
 Consider a non-degenerate eigenspectrum $\Lambda$. Then, the following statements are equivalent:
 \begin{enumerate}
 
 \item   $\Lambda \in \cP_1$
 
 \item  $ \Lambda \in \cP_2$
 
 \item  $\Lambda$ satisfies the $(S-1)$-pure unique pairing condition.
 
 \end{enumerate} 
 \end{corollary}
 \begin{proof}
 We have already noted that $(1) \rightarrow (2) $. $(3) \rightarrow (1)$ follows from Thm.~\ref{thm sub}. Finally, $ (2) \rightarrow (3)$ follows from Prop.~\ref{prop}, which we now proceed to show. 
 
 Assume the opposite. That is, $\Lambda$ does not satisfy the $(S-1)$-pure unique pairing condition. We have already noted (Rem.~\ref{remark Lambda_not_uniqpair}) that if $\Lambda$ does not satisfy the general unique pairing condition, then $\Lambda \notin \cP_2$. So, assume $\Lambda$ satisfies the general unique pairing condition. Then there exist at least two pairs $\{\lambda_i, \lambda_{2S-i+1}\}$ and $\{\lambda_j, \lambda_{2S-j+1}\}$ in the unique valid pairing such that $\lambda_i\lambda_{2S-i+1} >1$, $\lambda_j \lambda_{2S-j+1} >1$. This implies there are at least two thermal parameters $\nu_i,\nu_j $ associated with $\cD_\Lambda$ that satisfy $\nu_i>1,\nu_j>1 $. Thus, $\Lambda$ satisfies the conditions of Prop.~\ref{prop} and, hence $\Lambda \notin \cP_2$.
 \end{proof}

We finish this section by summarizing our results about $\cP_1(S)$ and $\cP_2(S)$ in several bullet points. More specifically, an eigenspectrum $\Lambda $ that satisfies the general unique pairing condition,
\begin{enumerate}
\item belongs to $\cP_1(S)$ and $\cP_2(S)$ if $\Lambda $ is non-degenerate and satisfies the $(S-1)$-pure unique pairing condition (Cor.~\ref{cor non-deg-class}).
\item belongs to $\cP_1(S)$ and $\cP_2(S)$ if at least $2S-1$ of the eigenvalues of $\Lambda$ are equal (Lem.~\ref{lemma 2S-1-deg eigenspectrum}).
\item does not belong to $\cP_1(S)$ and $\cP_2(S)$ if at least two thermal parameters $\nu_k$, $\nu_l$ associated with $\cD_\Lambda$ satisfy $\nu_k>1$, $\nu_l>1$, and the $4$-tuple $\{\nu_k e^{-2r_k},\nu_k^{2r_k},\nu_l e^{-2r_l},\nu_l e^{2r_l}\} \subseteq \Lambda$ has at most $2$-fold degeneracy (Prop.~\ref{prop}). 
\end{enumerate}

\section{Conclusion}\label{sec disc}
 
In this study we attempted to answer two related questions about the eigenspectra of quantum CMs: (1) to characterize the eigenspectra with the property that the corresponding set of quantum CMs form a single orbit under the action of $\mathrm{SpO}(2S,\mathbb{R})$ (Prob.~\ref{problem 1}), and (2) to characterize the eigenspectra with the property that the corresponding set of quantum CMs have the same set of thermal and squeezing parameters (Prob.~\ref{problem 2}). Our main result was solving these problems for the set of non-degenerate eigenspectra (Cor.~\ref{cor non-deg-class}), which constitute 	``almost all" eigenspectra. In particular, we found a non-trivial class of non-degenerate eigenspectra that satisfy both properties (Thm.~\ref{thm sub}). There is only one way to pair the elements of an eigenspectrum in this class such that the product of each pair is greater than or equal to $ 1$. Further, in this pairing, the product of only one of the pairs can be strictly greater than $1$. We named this set of pairing conditions on an eigenspectrum the ``$(S-1)$-pure unique pairing" condition. It is easy to show that the set of eigenspectra satisfying both properties is larger than the class we found, and we gave a class of examples (Lem.~\ref{lemma 2S-1-deg eigenspectrum}). 
 
 Further work is needed to fill the gaps in our findings, namely, to find the conditions on the eigenspectra that violate the non-degeneracy conditions in Thm.~\ref{thm sub}, as well as in Prop.~\ref{prop}, but possess either one or both of the properties in the posed problems. It is likely that the general proof methodology used to derive our results can also be successfully applied to the remaining class of eigenspectra. We conjecture that the complete answers to Probs.~\ref{problem 1} and~\ref{problem 2} are identical.

This study was in part inspired by the observation that, for a Gaussian state, the only features of its CM that that affect its total-photon-number distribution are the eigenvalues. This is a consequence of the fact that, in the phase-space formulation of quantum mechanics, the quasi-probability distribution, e.g. the Wigner function, of the total-photon-number operator is spherically symmetric. More generally, the eigenspectrum of the quantum CM captures the rotation-invariant features of the ``spread" of the Wigner function, while the symplectic eigenvalues capture the amount of spread in the subspaces associated with conjugate coordinates. The latter, for physical states, must be above certain thresholds due to Heisenberg's uncertainty principle. Informally speaking, we can view this study as an investigation of what the rotation-invariant spread of the Wigner function can tell us about the amount of spread as well as its uniformity, in the subspaces associated with conjugate coordinates. Future research directions can include considering much broader classes of eigenspectra of quantum CMs and attempting to determine the ranges in which the corresponding thermal and squeezing parameters must fall. 

\begin{acknowledgments}
The author acknowledges support from the Professional Research Experience Program (PREP) operated jointly by NIST and the University of Colorado. Discussions with Emanuel Knill contributed to thinking about the problem and to improving the formulation of the main statements. I would also like to thank Scott Glancy for helpful questions and suggestions to improve the presentation. The current version of the manuscript was significantly influenced by the helpful comments of the reviewer, for which I am very grateful.
\end{acknowledgments}

 \bibliography{covchar_rmp_non_deg_ver2.bib}

\begin{thebibliography}{14}%
\makeatletter
\providecommand \@ifxundefined [1]{%
 \@ifx{#1\undefined}
}%
\providecommand \@ifnum [1]{%
 \ifnum #1\expandafter \@firstoftwo
 \else \expandafter \@secondoftwo
 \fi
}%
\providecommand \@ifx [1]{%
 \ifx #1\expandafter \@firstoftwo
 \else \expandafter \@secondoftwo
 \fi
}%
\providecommand \natexlab [1]{#1}%
\providecommand \enquote  [1]{``#1''}%
\providecommand \bibnamefont  [1]{#1}%
\providecommand \bibfnamefont [1]{#1}%
\providecommand \citenamefont [1]{#1}%
\providecommand \href@noop [0]{\@secondoftwo}%
\providecommand \href [0]{\begingroup \@sanitize@url \@href}%
\providecommand \@href[1]{\@@startlink{#1}\@@href}%
\providecommand \@@href[1]{\endgroup#1\@@endlink}%
\providecommand \@sanitize@url [0]{\catcode `\\12\catcode `\$12\catcode
  `\&12\catcode `\#12\catcode `\^12\catcode `\_12\catcode `\%12\relax}%
\providecommand \@@startlink[1]{}%
\providecommand \@@endlink[0]{}%
\providecommand \url  [0]{\begingroup\@sanitize@url \@url }%
\providecommand \@url [1]{\endgroup\@href {#1}{\urlprefix }}%
\providecommand \urlprefix  [0]{URL }%
\providecommand \Eprint [0]{\href }%
\providecommand \doibase [0]{https://doi.org/}%
\providecommand \selectlanguage [0]{\@gobble}%
\providecommand \bibinfo  [0]{\@secondoftwo}%
\providecommand \bibfield  [0]{\@secondoftwo}%
\providecommand \translation [1]{[#1]}%
\providecommand \BibitemOpen [0]{}%
\providecommand \bibitemStop [0]{}%
\providecommand \bibitemNoStop [0]{.\EOS\space}%
\providecommand \EOS [0]{\spacefactor3000\relax}%
\providecommand \BibitemShut  [1]{\csname bibitem#1\endcsname}%
\let\auto@bib@innerbib\@empty
\bibitem [{\citenamefont {Simon}\ \emph {et~al.}(1994)\citenamefont {Simon},
  \citenamefont {Mukunda},\ and\ \citenamefont {Dutta}}]{simon1994quantum}%
  \BibitemOpen
  \bibfield  {author} {\bibinfo {author} {\bibfnamefont {R.}~\bibnamefont
  {Simon}}, \bibinfo {author} {\bibfnamefont {N.}~\bibnamefont {Mukunda}},\
  and\ \bibinfo {author} {\bibfnamefont {B.}~\bibnamefont {Dutta}},\ }\bibfield
   {title} {\bibinfo {title} {Quantum-noise matrix for multimode systems: U(n)
  invariance, squeezing, and normal forms},\ }\href
  {https://doi.org/10.1103/PhysRevA.49.1567} {\bibfield  {journal} {\bibinfo
  {journal} {Phys. Rev. A}\ }\textbf {\bibinfo {volume} {49}},\ \bibinfo
  {pages} {1567} (\bibinfo {year} {1994})}\BibitemShut {NoStop}%
\bibitem [{\citenamefont {{Arvind}}\ \emph {et~al.}(1995)\citenamefont
  {{Arvind}}, \citenamefont {Dutta}, \citenamefont {Mukunda},\ and\
  \citenamefont {Simon}}]{arvind1995real}%
  \BibitemOpen
  \bibfield  {author} {\bibinfo {author} {\bibnamefont {{Arvind}}}, \bibinfo
  {author} {\bibfnamefont {B.}~\bibnamefont {Dutta}}, \bibinfo {author}
  {\bibfnamefont {N.}~\bibnamefont {Mukunda}},\ and\ \bibinfo {author}
  {\bibfnamefont {R.}~\bibnamefont {Simon}},\ }\bibfield  {title} {\bibinfo
  {title} {The real symplectic groups in quantum mechanics and optics},\ }\href
  {https://doi.org/10.1007/BF02848172} {\bibfield  {journal} {\bibinfo
  {journal} {Pramana}\ }\textbf {\bibinfo {volume} {45}},\ \bibinfo {pages}
  {471} (\bibinfo {year} {1995})}\BibitemShut {NoStop}%
\bibitem [{\citenamefont {Williamson}(1936)}]{williamson1936on}%
  \BibitemOpen
  \bibfield  {author} {\bibinfo {author} {\bibfnamefont {J.}~\bibnamefont
  {Williamson}},\ }\bibfield  {title} {\bibinfo {title} {On the algebraic
  problem concerning the normal forms of linear dynamical systems},\
  }\href@noop {} {\bibfield  {journal} {\bibinfo  {journal} {American Journal
  of Mathematics}\ }\textbf {\bibinfo {volume} {58}},\ \bibinfo {pages}
  {141–163} (\bibinfo {year} {1936})}\BibitemShut {NoStop}%
\bibitem [{\citenamefont {Bhatia}\ and\ \citenamefont
  {Jain}(2015)}]{bhatia2015on}%
  \BibitemOpen
  \bibfield  {author} {\bibinfo {author} {\bibfnamefont {R.}~\bibnamefont
  {Bhatia}}\ and\ \bibinfo {author} {\bibfnamefont {T.}~\bibnamefont {Jain}},\
  }\bibfield  {title} {\bibinfo {title} {{On symplectic eigenvalues of positive
  definite matrices}},\ }\href {https://doi.org/10.1063/1.4935852} {\bibfield
  {journal} {\bibinfo  {journal} {Journal of Mathematical Physics}\ }\textbf
  {\bibinfo {volume} {56}},\ \bibinfo {pages} {112201} (\bibinfo {year}
  {2015})}\BibitemShut {NoStop}%
\bibitem [{\citenamefont {Son}\ \emph {et~al.}(2021)\citenamefont {Son},
  \citenamefont {Absil}, \citenamefont {Gao},\ and\ \citenamefont
  {Stykel}}]{son2021computing}%
  \BibitemOpen
  \bibfield  {author} {\bibinfo {author} {\bibfnamefont {N.~T.}\ \bibnamefont
  {Son}}, \bibinfo {author} {\bibfnamefont {P.-A.}\ \bibnamefont {Absil}},
  \bibinfo {author} {\bibfnamefont {B.}~\bibnamefont {Gao}},\ and\ \bibinfo
  {author} {\bibfnamefont {T.}~\bibnamefont {Stykel}},\ }\bibfield  {title}
  {\bibinfo {title} {Computing symplectic eigenpairs of symmetric
  positive-definite matrices via trace minimization and {R}iemannian
  optimization},\ }\href {https://doi.org/10.1137/21M1390621} {\bibfield
  {journal} {\bibinfo  {journal} {SIAM Journal on Matrix Analysis and
  Applications}\ }\textbf {\bibinfo {volume} {42}},\ \bibinfo {pages} {1732}
  (\bibinfo {year} {2021})}\BibitemShut {NoStop}%
\bibitem [{\citenamefont {Jain}\ and\ \citenamefont
  {Mishra}(2022)}]{jain2022deriv}%
  \BibitemOpen
  \bibfield  {author} {\bibinfo {author} {\bibfnamefont {T.}~\bibnamefont
  {Jain}}\ and\ \bibinfo {author} {\bibfnamefont {H.~K.}\ \bibnamefont
  {Mishra}},\ }\bibfield  {title} {\bibinfo {title} {Derivatives of symplectic
  eigenvalues and a {L}idskii type theorem},\ }\href
  {https://doi.org/10.4153/S0008414X2000084X} {\bibfield  {journal} {\bibinfo
  {journal} {Canadian Journal of Mathematics}\ }\textbf {\bibinfo {volume}
  {74}},\ \bibinfo {pages} {457–485} (\bibinfo {year} {2022})}\BibitemShut
  {NoStop}%
\bibitem [{\citenamefont {de~Gosson}(2001)}]{deGosson2001symplectic}%
  \BibitemOpen
  \bibfield  {author} {\bibinfo {author} {\bibfnamefont {M.}~\bibnamefont
  {de~Gosson}},\ }\bibfield  {title} {\bibinfo {title} {The symplectic camel
  and phase space quantization},\ }\href
  {https://doi.org/10.1088/0305-4470/34/47/313} {\bibfield  {journal} {\bibinfo
   {journal} {Journal of Physics A: Mathematical and General}\ }\textbf
  {\bibinfo {volume} {34}},\ \bibinfo {pages} {10085} (\bibinfo {year}
  {2001})}\BibitemShut {NoStop}%
\bibitem [{\citenamefont {Gromov}(1985)}]{gromov1985pseudo}%
  \BibitemOpen
  \bibfield  {author} {\bibinfo {author} {\bibfnamefont {M.}~\bibnamefont
  {Gromov}},\ }\bibfield  {title} {\bibinfo {title} {Pseudo holomorphic curves
  in symplectic manifolds},\ }\href {https://doi.org/10.1007/BF01388806}
  {\bibfield  {journal} {\bibinfo  {journal} {Inventiones mathematicae}\
  }\textbf {\bibinfo {volume} {82}},\ \bibinfo {pages} {307} (\bibinfo {year}
  {1985})}\BibitemShut {NoStop}%
\bibitem [{\citenamefont {Weedbrook}\ \emph {et~al.}(2012)\citenamefont
  {Weedbrook}, \citenamefont {Pirandola}, \citenamefont {Garcia-Patron},
  \citenamefont {Cerf}, \citenamefont {Ralph}, \citenamefont {Shapiro},\ and\
  \citenamefont {Lloyd}}]{weedbrook:qc2012a}%
  \BibitemOpen
  \bibfield  {author} {\bibinfo {author} {\bibfnamefont {C.}~\bibnamefont
  {Weedbrook}}, \bibinfo {author} {\bibfnamefont {S.}~\bibnamefont
  {Pirandola}}, \bibinfo {author} {\bibfnamefont {R.}~\bibnamefont
  {Garcia-Patron}}, \bibinfo {author} {\bibfnamefont {N.~J.}\ \bibnamefont
  {Cerf}}, \bibinfo {author} {\bibfnamefont {T.~C.}\ \bibnamefont {Ralph}},
  \bibinfo {author} {\bibfnamefont {J.~H.}\ \bibnamefont {Shapiro}},\ and\
  \bibinfo {author} {\bibfnamefont {S.}~\bibnamefont {Lloyd}},\ }\bibfield
  {title} {\bibinfo {title} {Gaussian quantum information},\ }\href
  {https://doi.org/10.1103/RevModPhys.84.621} {\bibfield  {journal} {\bibinfo
  {journal} {Rev. Mod. Phys.}\ }\textbf {\bibinfo {volume} {84}},\ \bibinfo
  {pages} {621} (\bibinfo {year} {2012})}\BibitemShut {NoStop}%
\bibitem [{\citenamefont {Serafini}(2017)}]{serafini2017quantum}%
  \BibitemOpen
  \bibfield  {author} {\bibinfo {author} {\bibfnamefont {A.}~\bibnamefont
  {Serafini}},\ }\href {https://doi.org/10.1201/9781315118727} {\emph {\bibinfo
  {title} {Quantum Continuous Variables: A Primer of Theoretical Methods}}}\
  (\bibinfo  {publisher} {CRC Press},\ \bibinfo {year} {2017})\BibitemShut
  {NoStop}%
\bibitem [{\citenamefont {Avagyan}\ \emph {et~al.}(2023)\citenamefont
  {Avagyan}, \citenamefont {Knill},\ and\ \citenamefont
  {Glancy}}]{avagyan2023multi}%
  \BibitemOpen
  \bibfield  {author} {\bibinfo {author} {\bibfnamefont {A.}~\bibnamefont
  {Avagyan}}, \bibinfo {author} {\bibfnamefont {E.}~\bibnamefont {Knill}},\
  and\ \bibinfo {author} {\bibfnamefont {S.}~\bibnamefont {Glancy}},\
  }\bibfield  {title} {\bibinfo {title} {Multi-mode {G}aussian state analysis
  with total-photon counting},\ }\href
  {https://doi.org/10.1088/1361-6455/ace175} {\bibfield  {journal} {\bibinfo
  {journal} {Journal of Physics B: Atomic, Molecular and Optical Physics}\
  }\textbf {\bibinfo {volume} {56}},\ \bibinfo {pages} {145501} (\bibinfo
  {year} {2023})}\BibitemShut {NoStop}%
\bibitem [{\citenamefont {Gerrits}\ \emph {et~al.}(2016)\citenamefont
  {Gerrits}, \citenamefont {Lita}, \citenamefont {Calkins},\ and\ \citenamefont
  {Nam}}]{gerrits2016super}%
  \BibitemOpen
  \bibfield  {author} {\bibinfo {author} {\bibfnamefont {T.}~\bibnamefont
  {Gerrits}}, \bibinfo {author} {\bibfnamefont {A.}~\bibnamefont {Lita}},
  \bibinfo {author} {\bibfnamefont {B.}~\bibnamefont {Calkins}},\ and\ \bibinfo
  {author} {\bibfnamefont {S.~W.}\ \bibnamefont {Nam}},\ }\bibinfo {title}
  {Superconducting transition edge sensors for quantum optics},\ in\ \href
  {https://doi.org/10.1007/978-3-319-24091-6_2} {\emph {\bibinfo {booktitle}
  {Superconducting Devices in Quantum Optics}}},\ \bibinfo {editor} {edited by\
  \bibinfo {editor} {\bibfnamefont {R.~H.}\ \bibnamefont {Hadfield}}\ and\
  \bibinfo {editor} {\bibfnamefont {G.}~\bibnamefont {Johansson}}}\ (\bibinfo
  {publisher} {Springer International Publishing},\ \bibinfo {address} {Cham},\
  \bibinfo {year} {2016})\ pp.\ \bibinfo {pages} {31--60}\BibitemShut {NoStop}%
\bibitem [{\citenamefont {Tao}()}]{taotopics}%
  \BibitemOpen
  \bibfield  {author} {\bibinfo {author} {\bibfnamefont {T.}~\bibnamefont
  {Tao}},\ }\href {https://books.google.com/books?id=Hjq\_JHLNPT0C} {\emph
  {\bibinfo {title} {Topics in Random Matrix Theory}}},\ Graduate studies in
  mathematics\ (\bibinfo  {publisher} {American Mathematical Soc.})\BibitemShut
  {NoStop}%
\bibitem [{\citenamefont {Bondy}\ and\ \citenamefont
  {Murty}(1976)}]{bondy1976graph}%
  \BibitemOpen
  \bibfield  {author} {\bibinfo {author} {\bibfnamefont {J.~A.}\ \bibnamefont
  {Bondy}}\ and\ \bibinfo {author} {\bibfnamefont {U.~S.~R.}\ \bibnamefont
  {Murty}},\ }\href@noop {} {\emph {\bibinfo {title} {Graph theory with
  applications}}}\ (\bibinfo  {publisher} {Macmillan London},\ \bibinfo {year}
  {1976})\BibitemShut {NoStop}%
\end{thebibliography}%

 \appendix

\section{Concerning the CMs of Pure Gaussian States}\label{app}

Here we show that if $\Lambda$ is an eigenspectrum of the quantum CM of a pure Gaussian state, then $\Lambda \in \cP_1$. This in turn implies that $\Lambda \in \cP_2$. The strategy used in the derivation is different from the approach of the main text that culminated in the proof of Thm.~\ref{thm sub}. We first highlight the following facts, using the concepts introduced in Sec.~\ref{sec prob_form}.
Assume $\Gamma$ is the quantum CM of a pure Gaussian state living in $S$ modes. Then, the thermal parameters $\nu_i$ of $\Gamma$ satisfy $\nu_i=1$ for all $i = 1,\hdots, S$. It follows that $\Gamma = A^TA$ for some symplectic matrix $A$. Thus, $\det(\Gamma) = 1$ since $\det(A)=1$. Writing $A = LQK$, where $K,L \in \mathrm{SpO}(2S,\mathbb{R})$ and $Q = \textrm{diag}[e^{-r_1},e^{r_1},\hdots, e^{-r_S},e^{r_S}]$, we obtain $ \Gamma = K^T Q^2 K$. 
\begin{lemma}\label{lem app}
Let $\Gamma_1$ and $\Gamma_2$ be two quantum CMs with the same eigenspectrum $\Lambda$. Assume $\Gamma_1$ and $\Gamma_2$ are the quantum CMs of pure Gaussian states living in $S$ modes. Then $\Gamma_1 = W^T\Gamma_2 W$ for some $W \in \mathrm{SpO}(2S,\mathbb{R})$.
\end{lemma}
\begin{proof}
We want to show that, given a quantum CM $\Gamma$ with the eigenspectrum $\Lambda$, any non-trivial ONS transformation $O^T \Gamma O$ by an ONS matrix $O$ results in a PDM that is not a quantum CM. Since $\Gamma' := O^T \Gamma O =O^T  K^T Q^2 K O$ is a non-trivial ONS transformation of $\Gamma$, there does not exist a $W' \in \mathrm{SpO}(2S,\mathbb{R})$ such that $\Gamma' = (W')^T Q^2 W'$. In other words, $\Gamma'$ is not a quantum CM of a pure Gaussian state. Hence, $\Gamma'$ has a decomposition $\Gamma' =  (A')^T T A'$, where $T \neq I_{2S}$ and $A'$ is a symplectic matrix. We assume $\Gamma'$ is a quantum CM and show that it leads to a contradiction. 

If $\Gamma'$ is a quantum CM, then $T>I_{2S}$ since we have excluded the case $T=I_{2S}$. Thus, 
\begin{equation}
\det(\Gamma') = \det((A')^T T A') = \det((A')^T  A') \det(T) = \det(T) >1.
\end{equation} 
But we also have
\begin{equation}
\det(\Gamma') = \det( O^T \Gamma O) = \det(O^T  O) \det(\Gamma) = \det(\Gamma) =1.
\end{equation} 
Thus, $\Gamma'$ is not a quantum CM.
\end{proof}

\end{document}